\documentclass[a4paper]{article}

\title{The Accuracy Smoothness Dilemma in Prediction: a Novel Multivariate M-SSA Forecast Approach}
\author{Marc Wildi}

\usepackage{a4wide}
\usepackage{amsmath}
\usepackage{amsfonts} 
\usepackage{amsthm}
\usepackage{amssymb}
\usepackage[utf8]{inputenc}
\usepackage{hyperref}
\usepackage{float}
\newtheorem{Proposition}{Proposition}
\newtheorem{Corollary}{Corollary}
\newtheorem{Theorem}{Theorem}

\usepackage{Sweave}
\begin{document}

\maketitle

\begin{abstract}
\noindent


Forecasting presents a complex estimation challenge, as it involves balancing multiple, often conflicting, priorities and objectives. Conventional forecast optimization methods typically emphasize a single metric—such as minimizing the mean squared error (MSE)—which may neglect other crucial aspects of predictive performance. To address this limitation, the recently developed Smooth Sign Accuracy (SSA) framework extends the traditional MSE approach by simultaneously accounting for sign accuracy, MSE, and the frequency of sign changes in the predictor. This addresses a fundamental trade-off—the so-called accuracy-smoothness (AS) dilemma—in prediction. We extend this approach to the multivariate M-SSA, leveraging the original criterion to incorporate cross-sectional information across multiple time series. As a result, the M-SSA criterion enables the integration of various design objectives related to AS forecasting performance, effectively generalizing conventional MSE-based metrics.  To demonstrate its practical applicability and versatility, we explore the application of the M-SSA in three primary domains: forecasting, real-time signal extraction (nowcasting), and smoothing. These case studies illustrate the framework’s capacity to adapt to different contexts while effectively managing inherent trade-offs in predictive modeling.


\end{abstract}

~
\newpage


\section{Introduction}

Forecasting presents a complex estimation challenge, as it necessitates the consideration of various, often conflicting, priorities and objectives. In this, we focus on accuracy and smoothness (AS). Accuracy pertains to the estimation of the future level of a time series while smoothness serves to regulate `noisy' changes.  
Ideally, these two dimensions could be jointly optimized, resulting in a predictor that closely approximates the true, albeit unobserved, future observation, thereby minimizing the incidence of spurious changes. 
However, the relationship between accuracy and smoothness presents a predictive dilemma; when appropriately formalized, an enhancement in one dimension inevitably leads to a compromise in the other. 
In this study, we propose a novel multivariate framework that enables users to manage and balance both accuracy and smoothness effectively. Furthermore, we extend existing methodologies to incorporate the proposed tradeoff, allowing for a customization of traditional `benchmark' predictors in terms of AS performances.\\

Smoothness of a predictor can be characterized by various concepts, such as its curvature, which is a geometric measure based on the squared second-order differences of the predictor. In contrast, the SSA proposed in Wildi (2024,2026) introduces a novel concept of smoothness that focuses on the expected duration between sign changes or zero crossings of a (stationary, zero-mean) predictor. The analysis of zero-crossings in time series data was first introduced by Rice (1944), who established a correlation between the autocorrelation function (ACF) of a zero-mean stationary Gaussian process and the expected number of zero crossings within a specified time interval. In the context of non-stationary time series, sign changes in the growth rate signify transitions between phases of expansion and contraction. These alternating periods of growth and contraction can be linked to the business cycle, provided that the fluctuations exhibit sufficient magnitude and duration. 
To be classified as representative of a business cycle, consecutive zero-crossings of the growth rate must be separated by intervals potentially spanning several years. This requirement for duration necessitates a smooth trajectory of the corresponding economic indicator. The relationship between the smoothness of a time series and the mean duration between consecutive sign changes—termed the holding time (HT) by Wildi (2024)—is formalized in the work of Rice (1944). The new M-SSA integrates accuracy within its objective function and smoothness through its holding time constraint, resulting in the M-SSA predictor being derived as the outcome of a constrained optimization process.
\\

The traditional mean-squared error (MSE) forecast paradigm predominantly emphasizes accuracy, often at the expense of smoothness. In certain applications, this focus on accuracy can lead to excessive noise leakage, thereby increasing the likelihood of false alarms, as demonstrated in Wildi (2024). To address this issue, McElroy and Wildi (2019) propose a forecasting approach based on a forecast trilemma; however, their solution does not take zero-crossings into account. Wildi (2024,2026) introduces the SSA optimization criterion, which explicitly regulates the HT of a predictor through a smoothing constraint derived from Rice's foundational work. We posit that regulating the rate of sign changes (in the growth rate), while simultaneously maintaining optimal tracking accuracy as measured by minimal MSE performance, provides a viable alternative to traditional filtering and smoothing methods, aligning with the decision-making and control logic pertinent during phases of alternating growth.Accordingly, we extend this approach  to multivariate prediction problems, referred to as M-SSA, with a focus on a stationary framework to clarify the core concepts. An extension of the univariate SSA to non-stationary integrated processes has been explored in Wildi (2026), where non-stationary maximal monotone and minimal curvature predictors are derived. These findings can be further generalized within the M-SSA framework. \\

Our applications aim to highlight some of the distinctive features of M-SSA in forecasting, real-time signal extraction, and smoothing. All examples can be replicated using an open-source M-SSA package, which includes an R package along with comprehensive instructions, practical use cases, and theoretical results, available at https://github.com/wiaidp/R-package-SSA-Predictor.git. A notable feature of our approach is the ability to customize benchmark predictors in terms of accuracy and smoothness (AS) performance. The M-SSA package demonstrates applications of this customization to various filtering techniques, including the Hodrick-Prescott (HP) filter (Hodrick and Prescott, 1997), Hamilton's regression filter (Hamilton, 2018), the Baxter-King (BK) bandpass filter (Baxter and King, 1999) and a `refined' Beveridge-Nelson design (Kamber, Morley and Wong (2024)).\\

Section \eqref{zc} briefly reviews the (univariate) SSA approach and introduces to the M-SSA criterion; Section \eqref{theorem_SSA} proposes solutions to the optimization problem; Section \eqref{examples} illustrates applications in forecasting, signal extraction and smoothing; finally, Section \eqref{conclusion} summarizes our main findings.

\section{M-SSA Criterion} \label{zc}

We provide a brief review of the univariate case and the various adaptations of the Smooth Sign Accuracy (SSA) optimization criterion introduced by Wildi (2024). Following this review, we derive the multivariate generalization M-SSA.

\subsection{Univariate SSA}

We consider a zero-mean stationary time series denoted as $x_t$ (processes with a non-vanishing mean are assumed to be centered) and a target variable $z_{t+\delta}$, where $ \delta \in \mathbb{Z}$, which depends on future values $x_{t-k}$ for $k < 0$. We then derive an optimal predictor $y_t$ for $z_{t+\delta}$ based on past values $x_{t-k}$ for $k \geq 0$ (a causal filter). This predictor is characterized by the property that the HT of $y_t$—defined as the mean duration between consecutive zero-crossings—can be specified by the user. For clarity and simplicity, we initially assume $x_t = \epsilon_t$ to be an independent and identically distributed white noise (WN) sequence. An extension of this approach to accommodate autocorrelated stationary processes will then be discussed.\\

Let 
\begin{equation}
z_t=\sum_{k=-\infty}^{\infty}\gamma_k x_{t-k},
\end{equation}
where $x_j=\epsilon_j, j \in \mathbb{Z}$, is WN. For the sake of clarity we may assume that $\epsilon_t$ is standardized. The sequence $\boldsymbol{\gamma} = (\gamma_k)$, where $k \in \mathbb{Z}$, is a real, square summable sequence, ensuring that $z_t$ constitutes a stationary process with zero mean and variance given by $\sum_{k=-\infty}^{\infty} \gamma_k^2$. We seek to construct a predictor defined as $y_t = \sum_{k=0}^{L-1} b_k \epsilon_{t-k}$ for the target $z_{t+\delta}$, where $\delta \in \mathbb{Z}$ and $b_k$ represents the coefficients of a one-sided causal filter of length $L$. This predictive framework is commonly categorized as fore-casting, now-casting, or back-casting, depending on whether $\delta > 0$, $\delta = 0$, or $ \delta < 0$, respectively.
For illustration,  consider the specific case where $\gamma_0 = 1$, $\gamma_1 = 0.5$, and $\gamma_k = 0$ for $k \notin {0, 1}$. In this scenario, $z_t = \epsilon_t + 0.5 \epsilon_{t-1}$ represents a moving-average process of order one. The task of one-step ahead forecasting for $z_t$ can be achieved by setting $\delta = 1$. The classic mean squared error (MSE) estimate for $z_{t+1}$ yields $y_t = 0.5 \epsilon_t$, resulting in the coefficients $b_0 = 0.5$ and $b_k = 0$ for $1 \leq k \leq L-1$\footnote{Furthermore, our target specification can also facilitate signal extraction problems, wherein the weights $\gamma_k$ are coefficients of a two-sided (potentially bi-infinite) filter.}.\\

The MSE predictor captures the level of future observations (or targets) optimally; however, it may produce forecasts that are 'noisy'. This phenomenon is evident in the previous moving average (MA(1)) example, where $y_t = 0.5\epsilon_t$ is white noise while the target $z_{t+1}$ is more persistent.  Similarly, in real-time signal extraction, the classic one-sided (nowcast) concurrent filter  tends to introduce markedly more `noise' relative to the acausal, two-sided target, as demonstrated in our subsequent examples. To mitigate this issue, Wildi (2024) proposes the following optimization problem:
\begin{eqnarray}\label{crit1}
\left.\begin{array}{cc}
&\max_{\mathbf{b}}\mathbf{b}'\boldsymbol{\gamma}_{\delta}\\
&\mathbf{b}'\mathbf{Mb}=l\rho_1\\
&\mathbf{b}'\mathbf{b}=l
\end{array}\right\}.
\end{eqnarray}
Criterion \eqref{crit1} is referred to as the SSA criterion, and its solution is denoted as SSA($\rho_1,\delta$). 
The constraints $\mathbf{b}'\mathbf{Mb}=l\rho_1$ and $\mathbf{b}'\mathbf{b}=l$ are termed the HT constraint and the length constraint, respectively. Smoothness of the predictor is obtained by specification of the hyperparameter $\rho_1$, which controls the first- order autocorrelation (ACF(1)) of the predictor\footnote{A formal link between ACF(1) and HT is established below.}. Here, $\mathbf{b}=(b_0,...,b_{L-1})'$ and $\boldsymbol{\gamma}_{\delta}=(\gamma_{\delta},...,\gamma_{\delta+L-1})'$ are column vectors of dimension $L$. The variable $l$ represents a scaling parameter, the specific choice of which will be discussed subsequently. This parameter is omitted from the criterion notation SSA($\rho_1,\delta$) due to its relatively minor importance in this context, see below for further reference. The autocovariance generating matrix $\mathbf{M}$ is defined as:
\[
\mathbf{M}=\left(\begin{array}{ccccccccc}0&0.5&0&0&0&...&0&0&0\\
0.5&0&0.5&0&0&...&0&0&0\\
...&&&&&&&&\\
0&0&0&0&0&...&0.5&0&0.5\\
0&0&0&0&0&...&0&0.5&0
\end{array}\right),
\]
which has dimension $L\times L$. The matrix $\mathbf{M}$ is structured such that $\mathbf{b}'\mathbf{Mb}=\sum_{k=1}^{L-1}b_{k-1}b_k$, representing the first-order autocovariance of $y_t$ under the specified assumptions. This formulation aims to optimize the predictor while controlling for the noise present in the forecasts. Under the assumption of WN, the classic MSE predictor is defined as $y_{t,MSE}:=\boldsymbol{\gamma}_{\delta}'\mathbf{x}_t$, where $\mathbf{x}_t:=(x_t,...,x_{t-(L-1)})'$. The weights $\boldsymbol{\gamma}_{\delta}$ can be derived as  a solution to the SSA criterion by setting $l:=\boldsymbol{\gamma}_{\delta}'\boldsymbol{\gamma}_{\delta}$ and $\rho_1:=\boldsymbol{\gamma}_{\delta}'\mathbf{M}\boldsymbol{\gamma}_{\delta}/\boldsymbol{\gamma}_{\delta}'\boldsymbol{\gamma}_{\delta}=:\rho_{MSE}$. However, the objective is to ensure that the SSA predictor $y_t := \mathbf{b}' \boldsymbol{\epsilon}_t$ exhibits reduced noise compared to $ y_{t,MSE}$ for which we can use the hyperparameter $\rho_1$ in the SSA criterion \eqref{crit1}. \\

\subsection{SSA Interpretations}\label{equi_s}

To streamline the terminology, we will refer to both $y_{t,MSE}$ and $\boldsymbol{\gamma}_{\delta}$ as the MSE predictor. Similarly, we will merge $y_t$ and $\mathbf{b}$ under the (M-)SSA designation, clarifying our intent when necessary. Under the assumption of white noise, $\boldsymbol{\gamma}_{\delta}$ represents the appropriate target $\boldsymbol{\gamma}$ in Criterion \eqref{crit1}, given that $\gamma_k$ is irrelevant for $k < \delta$ or $k > \delta + L - 1$. We will also assume $\boldsymbol{\gamma}_{\delta} \neq \mathbf{0}$. In this framework, the solution $\mathbf{b}_0$ to the SSA criterion can be interpreted as a constrained predictor for $z_{t+\delta}$. Alternatively, $\mathbf{b}_0$ can be viewed as a `smoother' for $y_{t,MSE}$, as discussed by Wildi (2024,2026).  Thus, the SSA criterion simultaneously addresses prediction and smoothing. Moreover, under the imposed length constraint, the objective function $\mathbf{b}' \boldsymbol{\gamma}_{\delta}$ is proportional to:
\[
\rho(y,z,\delta):=\mathbf{b}'\boldsymbol{\gamma}_{\delta}/\sqrt{l\boldsymbol{\gamma}'\boldsymbol{\gamma}}
\]
which represents the target correlation of $y_t$ with $z_{t+\delta}$, or to 
\[
\mathbf{b}'\boldsymbol{\gamma}_{\delta}/\sqrt{l\boldsymbol{\gamma}_{\delta}'\boldsymbol{\gamma}_{\delta}}
\]
the correlation of $y_t$ with $y_{t,MSE}$. Maximizing either of these functions consequently maximizes the other, allowing  Criterion \eqref{crit1} to be reformulated as:
\begin{eqnarray}\label{crit1e}
\left.\begin{array}{cc}
&\max_{\mathbf{b}}\rho(y,z,\delta)\\
&\rho(y)=\rho_1\\
&\mathbf{b}'\mathbf{b}=l
\end{array}\right. ~\textrm{or} \left.\begin{array}{cc}
&\max_{\mathbf{b}}\rho(y,y_{MSE},\delta)\\
&\rho(y)=\rho_1\\
&\mathbf{b}'\mathbf{b}=l
\end{array}\right\}.
\end{eqnarray}
Here, $\rho(y):=\mathbf{b}'\mathbf{Mb}/l=\mathbf{b}'\mathbf{Mb}/\mathbf{b}'\mathbf{b}$ denotes the ACF(1) of $y_t$. An increase in $\rho_1$ results in a stronger first-order ACF for the predictor, thus yielding a `smoother' trajectory for $y_t$, characterized by fewer zero-crossings. \\

Specifically, let $ht(y|\mathbf{b},i)$ be defined as $E[t_i-t_{i-1}]$, where $t_i$, $i\geq 1$ represents \emph{consecutive} zero-crossings of the process $y_t$\footnote{These crossings satisfy $t_{i-1}<t_i$ with the condition $t_1\geq L$ and $y_{t_{i}-1}y_{t_{i}}<0$ for all $i$. Additionally, it holds that $y_{t-1}y_{t}>0$ if $t_{i-1}<t<t_i$.}. Under the assumptions of stationarity, $ht(y|\mathbf{b},i)=ht(y|\mathbf{b})$ is invariant with respect to $i$. Furthermore, assuming Gaussianity of $x_t$, we have: 
\begin{equation}\label{ht}
ht(y|\mathbf{b})=\displaystyle{\frac{\pi}{\arccos(\rho(y))}},
\end{equation}
establishing a connection between the HT and the first-order ACF, see Wildi (2024). It is noteworthy that this relationship remains relatively robust even when deviations from Gaussian assumptions occur, as discussed in Wildi (2024). Given that the trigonometric function arccos is monotonic in the interval $]-1,1[$, we can reformulate the SSA criterion as follows: 
\begin{eqnarray}\label{crit1ee}
\left.\begin{array}{cc}
&\max_{\mathbf{b}}\rho(y,z,\delta)\\
&ht(y|\mathbf{b})=ht_1\\
&\mathbf{b}'\mathbf{b}=l
\end{array}\right\}.
\end{eqnarray}
Since correlations, signs, and zero-crossings are invariant to the scaling of $y_t$, the parameter $l$ in the length constraint can be treated as a nuisance parameter, which primarily serves to ensure uniqueness. We typically assume $l = 1$ unless otherwise stated. If necessary, `static' adjustments for level and scale can be applied post hoc, once a solution $ \mathbf{b} = \mathbf{b}(l)$ has been determined for an arbitrary $l$. This adjustment can be performed, for instance, by regressing the predictor on the target variable.\\

However, our primary focus lies in the `dynamic' aspects of the prediction problem, particularly concerning the target correlation and sign accuracy (SA). Specifically, let $SA(y_t) = P(z_{t+\delta}y_t>0)$ denote the probability of the target and predictor sharing the same sign. Under the Gaussian assumption, it follows that:
\[
SA(y_t)=0.5+\frac{\arcsin(\rho(y,z,\delta))}{\pi}
\]
thereby establishing a strictly monotonic (bijective) relationship between SA and target correlation. Consequently, the SSA criterion can be reformulated as:
\begin{eqnarray*}
&&\textrm{max}_{\mathbf{b}}SA(y_t)\\
&&ht(y|\mathbf{b})=ht_1\\
&&\mathbf{b}'\mathbf{b}=1,
\end{eqnarray*}
In conclusion, the SSA framework effectively reconciles MSE, SA and HT, i.e. smoothing, in a flexible and interpretable manner.

\subsection{Multivariate M-SSA}

Let $\mathbf{x}_t = (x_{1t}, \ldots, x_{nt})'$ represent a multivariate stationary process of dimension $n$, characterized by a (purely non-deterministic) Wold decomposition of the form
\begin{equation}\label{x_m_st}
\mathbf{x}_t = \sum_{k=0}^\infty \boldsymbol{\Xi}_k \boldsymbol{\epsilon}_{t-k},
\end{equation}
where $\boldsymbol{\Xi}_0 = \mathbf{I}_{n \times n}$ denotes the identity matrix, and $\boldsymbol{\epsilon}_t = (\epsilon_{1t}, \ldots, \epsilon_{nt})'$ is a sequence of multivariate WN with a variance-covariance matrix denoted as $\boldsymbol{\Sigma}$. The entries of $\boldsymbol{\Sigma}$ are represented by $\sigma_{ij}$, and we assume that its eigenvalues $\tilde{\sigma}_j$ for $j = 1, \ldots, n$ are strictly positive, indicating full rank\footnote{If $\boldsymbol{\Sigma}$ is originally rank deficient, we assume that redundant noise components have been eliminated, resulting in a lower-dimensional system that retains full rank.}. The eigenvectors of $\boldsymbol{\Sigma}$ are denoted by $\mathbf{v}_{\sigma k}$ for $k = 1, \ldots, n$. 
Next, we define a multivariate stationary process $\mathbf{z}_t$ of dimension $n$ as follows:
\[
\mathbf{z}_t=\sum_{|k|<\infty}\boldsymbol{\Gamma}_k\mathbf{x}_{t-k},
\] 
where $\boldsymbol{\Gamma}_k$ is a sequence of square-summable $n\times n$ dimensional filter matrices, with entries $(\gamma_{ijk})$ for $i, j = 1, \ldots, n$. The condition for square summability is expressed mathematically as:
\[
\sum_{|k| < \infty} \text{tr}(\boldsymbol{\Gamma}_k' \boldsymbol{\Gamma}_k) < \infty,
\]
where $\text{tr}(\cdot)$ denotes the trace operator. We focus on the estimation of $\mathbf{z}_{t+\delta}$, utilizing the $n-$dimensional predictor defined by
\[
\mathbf{y}_t = \sum_{k=0}^{L-1} \mathbf{B}_k \mathbf{x}_{t-k},
\]
with $n\times n$-dimensional filter matrices $\mathbf{B}_k$, $k=0,...,L-1$. To facilitate the theoretical analysis, we simplify our model by assuming $\mathbf{x}_t = \boldsymbol{\epsilon}_t$, with further extensions to autocorrelated processes addressed in the subsequent section. Let $b_{ijk}$, $1\leq i,j\leq n$, denote the entries of the matrix $\mathbf{B}_k$ for $k = 0, \ldots, L-1$. We define several vector notations as follows:
\begin{eqnarray*}
\boldsymbol{\epsilon}_{it}=(\epsilon_{it},\epsilon_{it-1},...,\epsilon_{it-(L-1)})'~&,&~\boldsymbol{\epsilon}_{\cdot t}=(\boldsymbol{\epsilon}_{1t}',...,\boldsymbol{\epsilon}_{nt}')'\\
\boldsymbol{\gamma}_{ij\delta}=(\gamma_{ij\delta},\gamma_{ij\delta+1},...,\gamma_{ij\delta+L-1})'~&,&\boldsymbol{\gamma}_{i\cdot\delta}=(\boldsymbol{\gamma}_{i1\delta}',\boldsymbol{\gamma}_{i2\delta}',...,\boldsymbol{\gamma}_{in\delta}')'\\
~\mathbf{b}_{ij}=(b_{ij0},b_{ij1},...,b_{ijL-1})'~&,&~\mathbf{b}_{i}=(\mathbf{b}_{i1}',\mathbf{b}_{i2}',...,\mathbf{b}_{in}')'\\
\boldsymbol{\gamma}_{\delta}=(\boldsymbol{\gamma}_{1\cdot\delta},\boldsymbol{\gamma}_{2\cdot\delta},...,
\boldsymbol{\gamma}_{n\cdot\delta})~&,&~\mathbf{b}=(\mathbf{b}_{1},\mathbf{b}_{2},...\mathbf{b}_{n}).
\end{eqnarray*} 
Here, $\boldsymbol{\gamma}_{\delta}$ and $\mathbf{b}$ are matrices of dimension $(L\cdot n)\times n$. The $i$-th column, $i=1,...,n$, contains the filter weights associated with the $i$-th target (or $i$-th predictor), and the filter weights corresponding to the $j=1,...,n$ series are stacked into this column vector, of total length $n\cdot L$.  The $i$-th column vector of $\boldsymbol{\gamma}_{\delta}$ (or $\mathbf{b}$) can be applied to the $L\cdot n$-dimensional $\boldsymbol{\epsilon}_{\cdot t}$, which stacks all $n$ series into a single long data vector. Specifically, we define 
\[
y_{ijt} := \mathbf{b}_{ij}'\boldsymbol{\epsilon}_{jt},
\]
which allows us to express the $i$-th predictor as:
\[
y_{it} = \mathbf{b}_i'\boldsymbol{\epsilon}_{\cdot t} = \sum_{j=1}^n y_{ijt}.
\] 
Similarly, the multivariate predictor vector can be written as:
\[
\mathbf{y}_t = \mathbf{b}'\boldsymbol{\epsilon}_{\cdot t}
\]
and the same applies to the target variables. Through orthogonal projection, the MSE predictor is derived as $\hat{\mathbf{y}}_{t,MSE} = \boldsymbol{\gamma}_{\delta}'\boldsymbol{\epsilon}_{\cdot t}$, with components $\hat{y}_{it,MSE}$, $i=1,...,n$.
Consider the Kronecker products defined as $\tilde{\mathbf{I}}:=\boldsymbol{\Sigma}\otimes\mathbf{I}_{L\times L}$ and $\tilde{\mathbf{M}}:=\boldsymbol{\Sigma}\otimes\mathbf{M}$, of the $n\times n$ dimensional matrix $\boldsymbol{\Sigma}$, the $L\times L$ dimensional identity  $\mathbf{I}_{L\times L}$ and the autocovariance generating matrix $\mathbf{M}$. We can express the following expectations: 
\begin{eqnarray}
E[z_{it+\delta}y_{it}]&=&E[\hat{y}_{it,MSE}y_{it}]=\boldsymbol{\gamma}_{i\cdot\delta}'\tilde{\mathbf{I}}\mathbf{b}_{i}\label{moment1}\\
E[\hat{z}_{it\delta}^2]&=&\boldsymbol{\gamma}_{i\cdot\delta}'\tilde{\mathbf{I}}\boldsymbol{\gamma}_{i\cdot\delta}\nonumber\\
E[y_{it}^2]&=&\mathbf{b}_{i}'\tilde{\mathbf{I}}\mathbf{b}_{i}\nonumber\\
E[y_{it-1}y_{it}]&=&\mathbf{b}_{i}'\tilde{\mathbf{M}}\mathbf{b}_{i}.\nonumber
\end{eqnarray}
The variance of $z_{it}$ is given by 
\[
E[z_{it}^2]=\sum_{|k|<\infty}\boldsymbol{\gamma}_{i k}'\boldsymbol{\Sigma}\boldsymbol{\gamma}_{i k},
\]
where $\boldsymbol{\gamma}_{i k}:=(\gamma_{i1k},...,\gamma_{ink})'$. Consequently, we propose the following multivariate M-SSA criterion for $i=1,...,n$: 
\begin{eqnarray}\label{mcrit1}
\textrm{max}_{\mathbf{b}_i}\boldsymbol{\gamma}_{i\cdot\delta}'\tilde{\mathbf{I}}\mathbf{b}_{i}\\
\mathbf{b}_{i}'\tilde{\mathbf{M}}\mathbf{b}_{i}=\rho_i\nonumber\\
\mathbf{b}_{i}'\tilde{\mathbf{I}}\mathbf{b}_{i}=1,\nonumber
\end{eqnarray}
for $i=1,...,n$. 
Under the imposed length constraint, the objective function is proportional to the correlation between the predictor and the $i$-th target (or  the $i$-th MSE predictor). Additionally, the HT constraint regulates the first-order ACF or, equivalently, the HT of the predictor. The relationships established between sign accuracy and target correlation, as well as between first-order ACF and HT, as discussed in the previous section, remain valid.



\subsection{Dependence}\label{ext_stat}

We now relax the WN hypothesis and assume a stationary process given by Equation \eqref{x_m_st}. Then, the target and predictor can be expressed formally as:
\begin{eqnarray*}
\mathbf{z}_t&=&\sum_{|k|<\infty}(\boldsymbol{\Gamma}\cdot\boldsymbol{\Xi})_k \boldsymbol{\epsilon}_{t-k}\\
\mathbf{y}_t&=&\sum_{j\geq 0} (\mathbf{B}\cdot\boldsymbol{\Xi})_j\boldsymbol{\epsilon}_{t-j}.
\end{eqnarray*} 
Here, $(\boldsymbol{\Gamma}\cdot\boldsymbol{\Xi})_k =\sum_{m\leq k}\boldsymbol{\Gamma}_m\boldsymbol{\Xi}_{k-m}$ and $(\mathbf{B}\cdot\boldsymbol{\Xi})_j=\sum_{m=0}^{\min(L-1,j)}\mathbf{B}_m \boldsymbol{\Xi}_{j-m} $ represent the convolutions of the sequences $\boldsymbol{\Gamma}_k$ and $\mathbf{B}_j$ with the Wold-decomposition $\boldsymbol{\Xi}_m$ of $\mathbf{x}_t$. For  given $(\mathbf{B}\cdot\boldsymbol{\Xi})_j$, $j=0,...,L-1$, the original coefficients $\mathbf{B}_{k}$ can be derived through deconvolution:
\begin{eqnarray}\label{con_inv}
\mathbf{B}_{k}&=&\left\{\begin{array}{cc}(\mathbf{B}\cdot\boldsymbol{\Xi})_{0}\boldsymbol{\Xi}_0^{-1}&,k=0\\
\left(  (\mathbf{B}\cdot\boldsymbol{\Xi})_{k}-\sum_{m=0}^{k-1}\mathbf{B}_m \boldsymbol{\Xi}_{j-m}\right)\boldsymbol{\Xi}_0^{-1}&,k=1,...,L-1
\end{array}\right.
\end{eqnarray}
In our parametrization, $\boldsymbol{\Xi}_0^{-1} = \mathbf{I}$ simplifies this expression. Denote the elements of the matrices as follows:
\[
(\boldsymbol{\Gamma}\cdot\boldsymbol{\Xi})_{lmk}
 \quad \text{and} \quad (\mathbf{B}\cdot\boldsymbol{\Xi})_{lmj}.
 \] 
We now assume $\delta\geq 0$\footnote{Similar but slightly more complex expressions would arise for $\delta < 0$ (backcasting), which are omitted here.} and $L$ to be sufficiently large to render $(\boldsymbol{\Gamma} \cdot \boldsymbol{\Xi})_{\delta+k} \approx 0$ and $(\mathbf{B} \cdot \boldsymbol{\Xi})_j \approx 0$ for $k,j \geq L$, so that  their contributions to the variances of $\hat{\mathbf{y}}_{t,MSE}$ and $\mathbf{y}_t$ are negligible. We define:
\begin{eqnarray*}
(\boldsymbol{\gamma}\cdot\boldsymbol{\xi})_{ij\delta}&=&\left((\boldsymbol{\Gamma}\cdot\boldsymbol{\Xi})_{ij\delta},(\boldsymbol{\Gamma}\cdot\boldsymbol{\Xi})_{ ij\delta+1}...,(\boldsymbol{\Gamma}\cdot\boldsymbol{\Xi})_{ij\delta+L-1}\right)'\\
(\boldsymbol{\gamma}\cdot\boldsymbol{\xi})_{i\delta}&=&\Big((\boldsymbol{\gamma}\cdot\boldsymbol{\xi})_{i1\delta}',(\boldsymbol{\gamma}\cdot\boldsymbol{\xi})_{i2\delta}',...,(\boldsymbol{\gamma}\cdot\boldsymbol{\xi})_{in\delta}'\Big)'\\
(\mathbf{b}\cdot\boldsymbol{\xi})_{ij}&=&\left((\mathbf{B}\cdot\boldsymbol{\Xi})_{ ij0},(\mathbf{B}\cdot\boldsymbol{\Xi})_{ij1},...,(\mathbf{B}\cdot\boldsymbol{\Xi})_{ijL-1}\right)'\\
(\mathbf{b}\cdot\boldsymbol{\xi})_{i}&=&\Big((\mathbf{b}\cdot\boldsymbol{\xi})_{i1}',(\mathbf{b}\cdot\boldsymbol{\xi})_{i2}',...,(\mathbf{b}\cdot\boldsymbol{\xi})_{in}'\Big)'
\end{eqnarray*}
It is important to note that if $\boldsymbol{\Xi}_k=\mathbf{0},~k>0$, resulting in $\mathbf{x}_t=\boldsymbol{\epsilon}_t$, then we have $(\boldsymbol{\gamma}\cdot\boldsymbol{\xi})_{i\delta}=\boldsymbol{\gamma}_{i\cdot\delta}$ and $(\mathbf{b}\cdot\boldsymbol{\xi})_{i}=\mathbf{b}_i$, as defined in the preceding section. The generalized M-SSA criterion can then be expressed as: 
\begin{eqnarray}\label{gen_stat_x}
\max_{(\mathbf{b}\cdot\boldsymbol{\xi})_i} (\boldsymbol{\gamma}\cdot\boldsymbol{\xi})_{i\delta}'\mathbf{\tilde{I}} (\mathbf{b}\cdot\boldsymbol{\xi})_i\\
(\mathbf{b}\cdot\boldsymbol{\xi})_i'\mathbf{\tilde{M}}(\mathbf{b}\cdot\boldsymbol{\xi})_i=\rho_i\nonumber\\
(\mathbf{b}\cdot\boldsymbol{\xi})_i'\mathbf{\tilde{I}}(\mathbf{b}\cdot\boldsymbol{\xi})_i=1\nonumber
\end{eqnarray}
for $i=1,...,n$. This formulation can be solved for $(\mathbf{b}\cdot\boldsymbol{\xi})_i$, as detailed in Theorem \eqref{lambda_mult} in the next section. The M-SSA coefficients $\mathbf{b}_i$ can subsequently be determined through deconvolution, as indicated in Equation \eqref{con_inv}. 
Wildi (2026) posits an extension of SSA to non-stationary integrated processes, wherein the HT constraint pertains to stationary differences of the predictor. While a corresponding extension to M-SSA is feasible based on findings by McElroy and Wildi (2020), the present discussion is confined to stationary processes, omitting a lengthier treatment of rank-defficient designs (cointegration) in the non-stationary case.\\

\textbf{Remark}: The target filter $\boldsymbol{\Gamma}_k$, for $|k|<\infty$,  and the MA-inversion $\boldsymbol{\Xi}_j$, for $j\geq 0$, are combined into the convolution $(\boldsymbol{\gamma}\cdot\boldsymbol{\xi})_{i\delta}$, for $i=1,...,n$, in the M-SSA criterion \eqref{gen_stat_x}. Due to the commutative property of convolution, this combination is invariant to the ordering of the components. Consequently, M-SSA treats the target filter and the Wold decomposition as essentially equivalent `filters', with the distinction being largely semantic. Our notation emphasizes this difference to distinguish between forecasting, smoothing, and signal extraction, which pertain to specific target formulations $\boldsymbol{\Gamma}_k$ for a given process $\boldsymbol{\Xi}_j$. For further clarification, see Section \eqref{examples}.

\subsection{Embedding, Smoothing and Efficient Frontier}\label{op_mod}

In the M-SSA criterion \eqref{gen_stat_x}, the terms $(\boldsymbol{\gamma}\cdot\boldsymbol{\xi})_{i\delta}$, $(\mathbf{b}\cdot\boldsymbol{\xi})_i$, $(\mathbf{b}\cdot\boldsymbol{\xi})_i$, $\mathbf{\tilde{M}}$ and $\mathbf{\tilde{I}}$ depend on the data-generating process (DGP), characterized by $\boldsymbol{\Sigma}$ and $\boldsymbol{\Xi}_k$, and on the forecast target $\mathbf{z}_t$, specified by $\boldsymbol{\Gamma}_k$. Accordingly, for $i=1,...n$, the optimal M-SSA coefficient vector $\mathbf{b}_i$ is a function of the DGP, the target specification, and the constraint parameter $\rho_i$. When $\rho_i=\rho_{i,MSE}$, where $\rho_{i,MSE}$ denotes the first-order ACF of the MSE predictor $\hat{y}_{it,MSE}$, the solution reduces to $\mathbf{b}_i=\boldsymbol{\gamma}_{i\cdot\delta}$,  thereby recovering the MSE predictor and rendering the HT constraint redundant (detailed in the subsequent section). Hence, the MSE predictor arises as a special case within the M-SSA framework (embedding).\\

The target  $z_{it}$  may be causal or acausal. With a causal target, M-SSA is interpreted as a \emph{smoother}; with an acausal target, it is interpreted as a \emph{predictor}. Standard multivariate one-step or $h$-step-ahead forecasts are obtained by setting 
\[
\boldsymbol{\Gamma}_k=\left\{\begin{array}{cc}\mathbf{I}_{n\times n}~,k=-h\\\mathbf{0}&,\textrm{otherwise}\end{array}\right.
\]
(see section \eqref{var1}). A nowcast ($h = 0$) or backcast ($h < 0$) is produced accordingly: the nowcast $y_{it}$ is smoother or noisier than $x_{it}$ depending on whether $\rho_i$ exceeds or falls short of the first-order ACF $\rho(x_i)$ of $x_{it}$ (see section \eqref{smoo_app}). More general prediction and smoothing tasks can be accommodated by allowing an arbitrary (generally square summable) sequence $\boldsymbol{\Gamma}_k$ in the target specification $\mathbf{z}_t$ (see section \eqref{se_app} for real-time trend extraction). In applications, the target specification can encode various research priorities by selecting the signals of interest.  Finally, any acausal prediction problem can be expressed as a causal smoothing problem by replacing the target $z_{it}$ in the objective with the  causal MSE predictor $\hat{y}_{it,MSE}$, as in \eqref{crit1e} and \eqref{moment1}. This reformulation is referred to as the equivalent causal M-SSA representation. \\

Under its primal specification \eqref{mcrit1}, or its extension \eqref{gen_stat_x}, M-SSA selects the predictor by maximizing correlation with the target, subject to a novel HT smoothing constraint. Conversely, the dual formulation detailed in the subsequent section characterizes the M-SSA predictor as the smoothest solution, with the largest HT, subject to a prescribed target correlation. Thus, for variable admissible $\rho_i$ in the HT constraint\footnote{It is assumed that the constraints under both primal and dual representations are feasible; technical details are provided in the subsequent section.}, M-SSA traces the efficient frontier linking accuracy (target correlation) and smoothness (zero-crossing rate) under the specified criterion.

\section{Solution}\label{theorem_SSA}



For clarity and simplicity, we assume that the process $\mathbf{x}_t = \boldsymbol{\epsilon}_t$ is WN, acknowledging that autocorrelated processes may be accommodated through the extensions proposed in the preceding section. Consequently, M-SSA is adopted throughout in baseline (MSE) extension mode.\\

We consider the orthonormal Fourier eigenvectors defined as $\mathbf{v}_j:=\left(\sin(k\omega_j)/\sqrt{\sum_{k=1}^Lsin(k\omega_j)^2}\right)_{k=1,...,L}$ of the matrix $\mathbf{M}$, associated with the eigenvalues  $\lambda_{j}=\cos(\omega_j)$ calculated at discrete Fourier frequencies $\omega_j=j\pi /(L+1)$, $j=1,...,L$, as detailed by Anderson (1975).  
We denote the solutions to the equation $\frac{\partial \rho(y_{i})}{\partial \mathbf{b}_i} = \mathbf{0}$ as stationary points of the first-order ACF $\rho(y_{i})$.

\begin{Proposition}\label{stationary_eigenvec}
Under the assumptions of white noise and a full-rank covariance matrix $\boldsymbol{\Sigma}$ as established in the previous section, the vector $\mathbf{b}_i$ of the M-SSA Criterion \eqref{mcrit1} constitutes a stationary point of the first-order autocorrelation $\rho(y_{i})$ if and only if $\mathbf{b}_{ij} \propto \mathbf{v}_k$ for $j = 1,...,n$, where $\mathbf{v}_k$ is an eigenvector of $\mathbf{M}$. In this scenario, the first-order autocorrelation is given by $\rho(y_{i}) = \lambda_k$, where $\lambda_k$ corresponds to the associated eigenvalue. 
The extremal values of $\rho(y_{i})$ are obtained as  $\rho_{min}(L)=\min_k\lambda_k=\cos(\pi L/(L+1))=-\cos(\pi /(L+1))$ and $\rho_{max}(L)=\max_k\lambda_k=\cos(\pi /(L+1))$. 
\end{Proposition}

\textbf{Proof}: For the sake of clarity, we assume that the vector $\mathbf{b}_i$ is constrained to the unit-ellipse defined by $\mathbf{b}_{i}'\tilde{\mathbf{I}}\mathbf{b}_{i}=1$, which implies that $\rho(y_{i})=\mathbf{b}_i'\mathbf{\tilde{M}}\mathbf{b}_i$. 
To find a stationary point of $\rho(y_{i})$, we set the  derivative of the Lagrangian $\mathfrak{L}=\mathbf{b}_i'\mathbf{\tilde{M}}\mathbf{b}_i-\lambda(\mathbf{b}_{i}'\tilde{\mathbf{I}}\mathbf{b}_{i}-1)$ to zero, leading to the equation $\left(\mathbf{\tilde{M}}-\lambda\mathbf{\tilde{I}}\right)\mathbf{b}_i=\mathbf{0}$. 
Given that $\boldsymbol{\Sigma}$ is assumed to be of full rank, we can manipulate the previous expression by multiplying it with $\boldsymbol{\Sigma}^{-1} \otimes \mathbf{I}_{L \times L} = \mathbf{\tilde{I}}^{-1}$, yielding 
\[
\left(\mathbf{I}_{n \times n} \otimes \mathbf{M} - \lambda \mathbf{I}_{nL \times nL}\right)\mathbf{b}_i = \mathbf{0}.
\]
This expression can be rewritten as $n$ distinct sub-space equations $\mathbf{M}\mathbf{b}_{ij} - \lambda \mathbf{b}_{ij} = \mathbf{0}$, for $j = 1, ..., n$. Therefore, $\lambda = \lambda_k$ must be an eigenvalue of $\mathbf{M}$ for some $k \in {1, ..., L}$. Furthermore, since the eigenvalues of $\mathbf{M}$ are distinct, the corresponding eigenvector $\mathbf{v}_k$ is uniquely determined, leading to the conclusion that $\mathbf{b}_{ij}\propto\mathbf{v}_k$, $j=1,...,n$. Additionally, we have 
\[
\rho(y_{ij})=\frac{\mathbf{b}_{ij}'\mathbf{Mb}_{ij}}{\mathbf{b}_{ij}'\mathbf{b}_{ij}}=\lambda\frac{\mathbf{b}_{ij}'\mathbf{b}_{ij}}{\mathbf{b}_{ij}'\mathbf{b}_{ij}}=\lambda_k.
\]
Since $\mathbf{b}_{ij}\propto\mathbf{v}_{k}$, we can express $y_{it}$ as follows
\[
y_{it}=\mathbf{b}_i'\boldsymbol{\epsilon}_{\cdot t}=\sum_{j=1}^n\mathbf{b}_{ij}'\boldsymbol{\epsilon}_{jt}=\mathbf{v}_{k}'\boldsymbol{\tilde{\epsilon}}_{t},
\]
where $\boldsymbol{\tilde{\epsilon}}_{t}=\sum_{j=1}^n a_j\boldsymbol{{\epsilon}}_{jt}$ is a linear combination of the components $\boldsymbol{{\epsilon}}_{jt}$. Given that  $\boldsymbol{\tilde{\epsilon}}_{t}$ is also WN, we conclude that $\rho(y_{i})=\rho(y_{ij})=\lambda_k$, as asserted. Lastly, since the  unit-ellipse is devoid of boundary-points, the extremal values $\rho_{min}(L)$ and $\rho_{max}(L)$ of the first-order ACF of $y_{it}$ must correspond to stationary points. Thus we have $\rho_{min}(L)=\min_k\lambda_k$ and $\rho_{max}(L)=\max_k\lambda_k$, which concludes the proof of the proposition.\hfill \qedsymbol{}\\


We will now derive the spectral decomposition of the MSE predictor $\boldsymbol{\gamma}_{i\cdot\delta}$. Let $\tilde{\lambda}_m$, $m=1,...,nL$, denote the eigenvalues of the matrix $\mathbf{\tilde{M}}$, with corresponding eigenvectors $\mathbf{\tilde{v}}_m$. Specifically, we have $\tilde{\lambda}_{(k-1)n+ j}=\lambda_k\tilde{\sigma}_j$ and $\mathbf{\tilde{v}}_{(k-1)n+ j}=\mathbf{v}_{\sigma j}\otimes\mathbf{v}_k$, where  $\lambda_k$, $\tilde{\sigma}_j$ are the eigenvalues and  $\mathbf{v}_k$, $\mathbf{v}_{\sigma j}$ are the eigenvectors of $\mathbf{M}$ and $\boldsymbol{\Sigma}$, respectively, as established by Horn and Johnson (1991). We will index the eigenvalues and eigenvectors of $\tilde{\mathbf{M}}$ using the notation $\tilde{\lambda}_{kj}$ and $\mathbf{\tilde{v}}_{kj}$, where the two-dimensional index $kj$ refers to the product $\lambda_k\tilde{\sigma}_j$. It is assumed that the eigenvalues $\lambda_k$  are ordered in decreasing magnitude, from largest to smallest. Furthermore, we assume that the eigenvectors $\mathbf{\tilde{v}}_{kj}$ are normalized to form an orthonormal basis of $\mathbb{R}^{nL}$. The eigenvalues of  $\mathbf{\tilde{I}}$ are $\tilde{\sigma}_j$, each with multiplicity $L$, for $j=1,...,n$ (if $\tilde{\sigma}_j$ are not pairwise different then the multiplicities will be accordingly larger), and the eigenvectors $\mathbf{\tilde{v}}_{kj}$ of $\mathbf{\tilde{M}}$ are also eigenvectors of $\mathbf{\tilde{I}}$. By organizing the orthonormal basis  $\mathbf{\tilde{v}}_{kj}$ into the $nL\times nL$-dimensional matrix $\mathbf{\tilde{V}}$, we can express the MSE predictor $\boldsymbol{\gamma}_{i\cdot\delta}$, for $i=1,...,n$, in the form of the spectral decomposition:  
\begin{eqnarray}\label{specdec_mult}
\boldsymbol{\gamma}_{i\cdot\delta}=\mathbf{\tilde{V}}\mathbf{w}_i=\sum_{k=1}^{L}\sum_{j=1}^n w_{ikj}\mathbf{\tilde{v}}_{kj},
\end{eqnarray}
where the column vector $\mathbf{w}_i=(w_{ikj})_{k=1,...,L, j=1,...,n}$ of length $nL$  represents the spectral weights of the decomposition of $\boldsymbol{\gamma}_{i\cdot\delta}$. 
The predictor $\boldsymbol{\gamma}_{i\cdot\delta}$ is said to possess \emph{complete spectral support} if the condition $\sum_{j=1}^n|w_{ikj}|\neq 0$ holds for $k=1,...,L$. This definition generalizes the concept proposed by Wildi (2024) in the univariate context, where $n=1$ and it is required that $w_{ikj}=w_k\neq 0$ for $1\leq k\leq L$.

\begin{Proposition}\label{bound12}
Consider   the M-SSA criterion defined in Equation \eqref{mcrit1}, under the assumptions of WN and a full-rank covariance matrix $\boldsymbol{\Sigma}$, with $L\geq 2$. The following assertions hold:
\begin{itemize}
\item If $|\rho_i|>\rho_{max}(L)$ (exterior point) for the HT constraint of the $i$-th target, the corresponding optimization problem does not yield a solution. 
\item If $\rho_i=\lambda_1=\rho_{max}(L)$ (first boundary case) and  $\sum_{j=1}^n|w_{i1j}|\neq 0$,  then the M-SSA predictor is given by
\[
\mathbf{b}_i=\frac{\sqrt{l}}{\sqrt{\sum_{j=1}^n\tilde{\sigma}_jw_{i1j}^2}}\sum_{j=1}^n w_{i1j}\tilde{\mathbf{v}}_{1j}\neq \mathbf{0},
\]
where $w_{i1j}$ and $\mathbf{\tilde{v}}_{1j}$ represent the spectral weights and eigenvectors, as specified in Equation \eqref{specdec_mult}.  
\item Similarly, if $\rho_i=\lambda_L=-\rho_{max}(L)$ (second boundary case) and  $\sum_{j=1}^n|w_{iLj}|\neq 0$, then the M-SSA solution is expressed as
\[
\mathbf{b}_i=\frac{\sqrt{l}}{\sqrt{\sum_{j=1}^n\tilde{\sigma}_jw_{iLj}^2}}\sum_{j=1}^n w_{iLj}\tilde{\mathbf{v}}_{Lj}\neq \mathbf{0}.
\]
\end{itemize}
\end{Proposition}

\textbf{Proof}:  The proof of the first assertion is derived from Proposition \eqref{stationary_eigenvec}. For the second assertion, we assume $i=1$ and analyze the M-SSA Criterion \eqref{mcrit1} defined as:
\begin{eqnarray}
\textrm{max}_{\mathbf{b}_{1}}~\boldsymbol{\gamma}_{1\cdot\delta}'\tilde{\mathbf{I}}\mathbf{b}_{1}&&\label{crit_ssa_proof}\\  
\mathbf{b}_{1}'\mathbf{\tilde{M}}\mathbf{b}_{1}&=&\rho_1\label{hyperbola_mult}\\
\mathbf{b}_{1}'\mathbf{\tilde{I}}\mathbf{b}_{1}&=&1\label{unit_circle_mult},
\end{eqnarray}
where we assume $l=1$ in the length constraint. Applying spectral decomposition to $\mathbf{b}_1$ yields:  
\begin{eqnarray}\label{specdecdecb_multi}
\mathbf{b}_1:=\sum_{k=1}^{L}\sum_{j=1}^n\alpha_{kj}\mathbf{\tilde{v}}_{kj},
\end{eqnarray}
where the two-dimensional index corresponds to the eigenvalues $\tilde{\lambda}_{kj}=\lambda_k\tilde{\sigma}_j$ of $\mathbf{\tilde{M}}$, with $\lambda_k$ ordered from largest to smallest. The variable $\alpha_{kj}=\alpha_{1kj}$ depends on $i=1$, and we omit the complex three dimensional indexing for clarity. Since  $\tilde{\mathbf{v}}_{kj}$ serves as an (normalized) eigenvector of $\mathbf{\tilde{I}}$, with strictly positive eigenvalue $\tilde{\sigma}_j>0$, the length constraint can be reformulated as follows: 
\begin{eqnarray}\label{length_const}
1=\mathbf{b}_{i}'\tilde{\mathbf{I}}\mathbf{b}_{i}=\sum_{j=1}^n\tilde{\sigma}_j\sum_{k=1}^{L}\alpha_{kj}^2,
\end{eqnarray} 
which we refer to as the elliptic (length) constraint. In a similar manner,  the HT constraint (assuming $l=1$) yields the following expression:
\begin{eqnarray*}
\rho_1=\mathbf{b}_1'\mathbf{\tilde{M}}\mathbf{b}_1=\sum_{j=1}^n\tilde{\sigma}_j\sum_{k=1}^L \lambda_k \alpha_{kj}^2
\end{eqnarray*}
which we designate as the hyperbolic HT constraint. We can solve this equation for any $\alpha_{kj}$, at least if  $\lambda_k\neq 0$. Let $\lambda_2\neq 0$\footnote{Since $\lambda_k=\cos(k\pi/(L+1))$, it follows that $\lambda_2\neq 0$ when $L>3$, which is virtually always the case in practical applications. For $L=3$ a similar proof applies when selecting $\lambda_3\neq 0$ instead of $\lambda_2$.} so that 
\begin{eqnarray}\label{alpha2}
\alpha_{21}^2&=&\frac{\rho_1}{\lambda_2\tilde{\sigma}_1}-\sum_{(k,j)\neq(2,1)}\frac{\lambda_k\tilde{\sigma}_j}{\lambda_2\tilde{\sigma}_1}\alpha_{kj}^2,
\end{eqnarray}
where the notation $\sum_{(k,j)\neq(2,1)}$ indicates summation  over all combinations of $(k,j)$ excluding the single pairing $(k,j)=(2,1)$. Similarly, we can express $\alpha_{11}$ in the context of the elliptic constraint as follows:
\begin{eqnarray*}
\alpha_{11}^2&=&\frac{1}{\tilde{\sigma}_1}-\frac{1}{\tilde{\sigma}_1}\sum_{(k,j)\neq(1,1)}\tilde{\sigma}_j\alpha_{kj}^2,
\end{eqnarray*}
which can be rewritten as
\begin{eqnarray*}
\alpha_{11}^2&=&\frac{1}{\tilde{\sigma}_1}-\frac{1}{\tilde{\sigma}_1}\sum_{(k,j)\notin\{(2,1),(1,1)\}}\tilde{\sigma}_j\alpha_{kj}^2-\left(\frac{\rho_1}{\lambda_2\tilde{\sigma}_1}-\sum_{(k,j)\neq(2,1)}\frac{\lambda_k\tilde{\sigma}_j}{\lambda_2\tilde{\sigma}_1}\alpha_{kj}^2\right).
\end{eqnarray*}
Here, we substitute Equation \eqref{alpha2}. This formulation characterizes the intersection of elliptic and hyperbolic constraints: if this intersection is empty, then the M-SSA problem lacks a solution. Assuming $\rho_1=\lambda_1=\max_k\lambda_k$ and isolating $\tilde{\sigma}_1\alpha_{11}^2$ to the left side of the last equation, we derive:
\begin{eqnarray}
\tilde{\sigma}_1\alpha_{11}^2&=&\frac{\lambda_2-\rho_1}{\lambda_2-\lambda_1}-\sum_{(k,j)\notin\{(2,1),(1,1)\}}\tilde{\sigma}_j\frac{\lambda_2-\lambda_k}{\lambda_2-\lambda_1}\alpha_{kj}^2\label{solve_alpha_11}\\
&=&1-\sum_{(k,j)\notin\{(2,1),(1,1:n)\}}\tilde{\sigma}_j\frac{\lambda_2-\lambda_k}{\lambda_2-\lambda_1}\alpha_{kj}^2-\sum_{j=2}^n \tilde{\sigma}_j\frac{\lambda_2-\lambda_1}{\lambda_2-\lambda_1}\alpha_{1j}^2\nonumber\\
&=&1-\sum_{(k,j)\notin\{(2,1),(1,1:n)\}}\tilde{\sigma}_j\frac{\lambda_2-\lambda_k}{\lambda_2-\lambda_1}\alpha_{kj}^2-\sum_{j=2}^n \tilde{\sigma}_j\alpha_{1j}^2.\nonumber
\end{eqnarray}
where the notation $\sum_{(k,j)\notin\{(2,1),(1,1:n)\}}$ signifies that the sum extends over all combinations of $(k,j)$ excluding $k=1$ and the single pairing $(k=2,j=1)$. We can rewrite this expression as:
\begin{eqnarray}\label{ito_v}
\sum_{j=1}^n \tilde{\sigma}_j\alpha_{1j}^2&=&1-\sum_{(k,j)\notin\{(2,1),(1,1:n)\}}\tilde{\sigma}_j\frac{\lambda_2-\lambda_k}{\lambda_2-\lambda_1}\alpha_{kj}^2.
\end{eqnarray}
Since $\tilde{\sigma}_j>0$ (full rank) and $\lambda_k<\lambda_j$, if $k>j$, we deduce 
\[
\tilde{\sigma}_j\frac{\lambda_2-\lambda_k}{\lambda_2-\lambda_1}\left\{\begin{array}{cc}=0&k=2\\<0&k>2\end{array}\right..
\]
It follows that $\alpha_{kj}=0$ for all pairs $(k,j)$ with $k>2$ since otherwise the condition  $\sum_{j=1}^n \tilde{\sigma}_j\alpha_{1j}^2>1$ implied by Equation \eqref{ito_v} would contradict the elliptic (length) constraint specified in Equation \eqref{length_const}. Consequently, we deduce: 
\begin{eqnarray}\label{reduced_rank_el_eq}
\sum_{j=1}^n \tilde{\sigma}_j\alpha_{1j}^2=1\quad \text{and} \quad \alpha_{kj}=0, k\neq 1,
\end{eqnarray} 
where all zero-constraints are attributable to the elliptic constraint when $\boldsymbol{\Sigma}$ is of full rank. It is important to note that if $\rho_1 > \lambda_1$, then the expression 
\[
\frac{\lambda_2-\rho_1}{\lambda_2-\lambda_1}>1
\]
in Equation \eqref{solve_alpha_11} would conflict with the elliptic constraint, thereby providing an alternative proof of the first claim (exterior point).  In summary, when $\rho_1=\lambda_1$ (the first boundary case), the intersection of the ellipse and hyperbola manifests as a  `reduced-rank' $(n-1)$-dimensional ellipse as delineated by Equation \eqref{reduced_rank_el_eq}. The spectral decomposition expressed in Equation \eqref{specdecdecb_multi} then simplifies to 
\[
\mathbf{b}_1=\sum_{j=1}^n \alpha_{1j}\tilde{\mathbf{v}}_{1j}.
\]
Substituting this expression into the objective function represented in equation \eqref{crit_ssa_proof} results in 
\[
\boldsymbol{\gamma}_{1\cdot\delta}'\tilde{\mathbf{I}}\mathbf{b}_{1}=\sum_{j=1}^n w_{11j}\tilde{\sigma}_j\alpha_{1j},
\]
leading to the corresponding boundary (`reduced-rank') M-SSA problem defined as:
\begin{eqnarray*}
\max_{\alpha_{1j}}\sum_{j=1}^n w_{11j}\tilde{\sigma}_j\alpha_{1j}~~\textrm{subject~to}~~\sum_{j=1}^n \tilde{\sigma}_j\alpha_{1j}^2=1.
\end{eqnarray*}
We can now specify the Lagrangian function 
\[
\mathcal{L}:=\sum_{j=1}^n w_{11j}\tilde{\sigma}_j\alpha_{1j}-\lambda_{\mathcal{L}}\left(\sum_{j=1}^n \tilde{\sigma}_j\alpha_{1j}^2-1\right)
\]
where $\lambda_{\mathcal{L}}$ is the Lagrange multiplier. The corresponding Lagrangian equations yield
\[
\tilde{\sigma}_jw_{11j}-2\lambda_{\mathcal{L}}\tilde{\sigma}_j\alpha_{1j}=0~,~j=1,...,n
\]
which leads to the expression
\[\alpha_{1j}=w_{11j}/(2\lambda_{\mathcal{L}}).
\] 
Thus we can express
\[
\mathbf{b}_1=\frac{1}{2\lambda_{\mathcal{L}}}\sum_{j=1}^n w_{11j}\tilde{\mathbf{v}}_{1j}\neq \mathbf{0}
\]
with the proportionality term 
\[
\frac{1}{2\lambda_{\mathcal{L}}}:=\frac{1}{\sqrt{\sum_{j=1}^n\tilde{\sigma}_jw_{11j}^2}}
\]
ensuring compliance with the elliptic length constraint (assuming $l=1$). A proof of the last assertion (second boundary case) follows a similar reasoning, by symmetry.\hfill \qedsymbol{}\\

\textbf{Remark}:
If $|\rho_i|<\rho_{max}(L)$ (an interior point), then the intersection of the ellipse and hyperbola results in a subspace of dimension $Ln-2$ in contrast to the boundary scenario described by Equation \eqref{reduced_rank_el_eq}, which has dimension $n-1$. We will subsequently derive the corresponding M-SSA predictor for the `interior point' case, which is of primary interest in practical applications. For the purposes of this discussion, we will assume $Ln > 2$ to ensure that the optimization problem remains non-trivial (since otherwise the solution would be determined by the two constraints).

\begin{Theorem}\label{lambda_mult}
Assuming $Ln > 2$, we define a set of regularity conditions for the M-SSA optimization problem \eqref{mcrit1} under the stipulated assumptions of WN, full-rank covariance matrix $\boldsymbol{\Sigma}$ and $\boldsymbol{\gamma}_{\delta}\neq\mathbf{0}$:
\begin{enumerate}
\item $\rho_i\neq \rho_{MSE}$ (non-degenerate case).
\item $|\rho_i|<\rho_{max}(L)$ (interior point).
\item The MSE-estimate $\boldsymbol{\gamma}_{i\cdot\delta}$ exhibits complete spectral support (completeness).
\end{enumerate}
Then the following results hold: 
\begin{enumerate}
\item \label{ass1_mult} Under the above regularity assumptions, the solution $\mathbf{b}_{i}$ to the M-SSA Criterion \eqref{mcrit1} can be expressed in a general parametric form as follows:
\begin{eqnarray}\label{ssa_mult}
\mathbf{b}_{i}=D_i\mathbf{N}_i^{-1}\mathbf{\tilde{I}}\boldsymbol{\gamma}_{i\cdot\delta}=D_i\sum_{k=1}^L\frac{1}{(2\lambda_{k}-\nu_i)}\left(\sum_{j=1}^n w_{ikj}\mathbf{\tilde{v}}_{kj}\right),
\end{eqnarray}
where $D_i\neq 0$, $\nu_i\in \mathbb{R}-\{2\lambda_k|k=1,...,L\}$ and $\mathbf{N}_i:=2\mathbf{\tilde{M}}-\nu_i\mathbf{\tilde{I}}$ is an invertible and symmetric $nL\times nL$ matrix. Additionally,  the scaling $D_i=D_i(\nu_i,l)$ is contingent on $\nu_i$ and the length constraint with its sign determined by the requirement  for a positive objective function, specifically $\boldsymbol{\gamma}_{i\cdot\delta}'\tilde{\mathbf{I}}\mathbf{b}_{i}>0$. 

\item \label{ass3_mult}The first-order ACF $\rho(\nu_i)$ of $y_{it}$ based on $\mathbf{b}_{i}=\mathbf{b}_{i}(\nu_i)$ in Equation \eqref{ssa_mult}, is given by:
\begin{eqnarray}\label{spec_dec_rho_mult}
\rho(\nu_i):=\frac{\mathbf{b}_{i}(\nu_i)'\mathbf{\tilde{M}}\mathbf{b}_{i}(\nu_i)}{\mathbf{b}_{i}(\nu_i)'\mathbf{\tilde{I}}\mathbf{b}_{i}(\nu_i)}=\frac{\sum_{k=1}^{L}\lambda_k\frac{1}{(2\lambda_k-\nu_i)^2}\left(\sum_{j=1}^{n}\tilde{\sigma}_jw_{ikj}^2\right)}{\sum_{k=1}^{L}\frac{1}{(2\lambda_k-\nu_i)^2}\left(\sum_{j=1}^{n}\tilde{\sigma}_jw_{ikj}^2\right)}
\end{eqnarray}
Any $\rho_i$ such that $|\rho_i|<\rho_{max}(L)$ is permissible under the HT constraint, implying the existence of $\nu_i$ such that $\rho(\nu_i)=\rho_i$.

\item \label{ass4_mult} The derivative  $d \rho(\nu_i)/d\nu_i$ is strictly negative for $\nu_i\in\{\nu||\nu|>2\rho_{max}(L)\}$.  Furthermore, it holds that 
\[
\textrm{max}_{\nu_i<-2\rho_{max}(L)}\rho(\nu_i)=\textrm{min}_{\nu_i>2\rho_{max}(L)}\rho(\nu_i)=\rho_{i,MSE},
\] 
where $\rho_{i,MSE}=\lim_{|\nu_i|\to\infty}\rho(\nu_i)$ represents the first-order ACF of the MSE predictor  $y_{it,MSE}$. Finally, for   any $\nu_{1i} > 2\rho_{max}(L)$ and $\nu_{2i}< -2\rho_{max}(L)$, it follows that $\rho_i(\nu_{1i})>\rho_i(\nu_{2i})$.

\item \label{ass5} For $\nu_i\in\{\nu||\nu|>2\rho_{max}(L)\}$ the derivatives of the objective function and the first-order ACF  as a function of $\nu_i$ are interconnected by the following relation:
\begin{eqnarray}\label{ficcc}
-\textrm{sign}(\nu_i)\frac{d\rho(y_i(\nu_i),z_i,\delta))}{d\nu_i}=\frac{\sqrt{\boldsymbol{\gamma}_{i\cdot\delta}'\mathbf{\tilde{I}}'\mathbf{N}_i^{-2}\mathbf{\tilde{I}}\boldsymbol{\gamma}_{i\cdot\delta}}}{\sqrt{\boldsymbol{\gamma}_{i\cdot\delta}'\mathbf{\tilde{I}}'\mathbf{\tilde{I}}\boldsymbol{\gamma}_{i\cdot\delta}}}\frac{d\rho(\nu_i)}{d\nu_i}<0,
\end{eqnarray} 
where $\mathbf{N}_i^{-2}:=\mathbf{N}_i^{-1}~'\mathbf{N}_i^{-1}=\left(\mathbf{N}_i^{-1}\right)^2$, $y_{it}(\nu_i)$ denotes the output of the filter with weights $\mathbf{b}_{i}$ determined by Equation \eqref{ssa_mult}, and $\rho(y_i(\nu_i),z_i,\delta)$
corresponds to the target correlation and objective function.

\end{enumerate}
\end{Theorem}

\textbf{Proof}: Consider the Lagrangian defined as 
\begin{eqnarray*}
\mathcal{L}_i:=\boldsymbol{\gamma}_{i\cdot\delta}'\tilde{\mathbf{I}}\mathbf{b}_{i}-\lambda_{i1}(\mathbf{b}_{i}'\mathbf{\tilde{I}}\mathbf{b}_{i}-1)-\lambda_{i2}(\mathbf{b}_{i}'\mathbf{\tilde{M}}\mathbf{b}_{i}-\rho_i),
\end{eqnarray*}
where $\lambda_{i1},\lambda_{i2}$ denote the Lagrange multipliers and we assume $l=1$ in the length constraint. Since the intersection of the ellipse (representing the length constraint) and hyperbola (representing the HT constraint) is devoid of boundary points, the solution $\mathbf{b}_{i}$ to the M-SSA problem must satisfy the Lagrangian equations given by
\[
\tilde{\mathbf{I}}\boldsymbol{\gamma}_{i\cdot\delta}=\lambda_{i1}(\mathbf{\tilde{I}}+\mathbf{\tilde{I}}')\mathbf{b}_{i}+\lambda_{i2} (\mathbf{\tilde{M}}+\mathbf{\tilde{M}}')\mathbf{b}_{i}=\lambda_{i1} 2\mathbf{\tilde{I}}\mathbf{b}_{i}+\lambda_{i2} 2\mathbf{\tilde{M}}\mathbf{b}_{i}.
\]
The second regularity assumption (non-degenerate case) implies that the HT constraint, as expressed in  Equation \eqref{hyperbola_mult}, is `active', leading to $\lambda_{i2}\neq 0$.  Dividing through by $\lambda_{i2}$ yields 
\begin{eqnarray}
D_i\tilde{\mathbf{I}}\boldsymbol{\gamma}_{i\cdot\delta}&=& \mathbf{N}_i\mathbf{b}_i\label{diff_non_hom_matrix_mult}\\
\mathbf{N}_i&:=&(2\mathbf{\tilde{M}}-\nu_i\mathbf{\tilde{I}}),\label{labelNu_mult}
\end{eqnarray}
where $D_i=1/\lambda_{i2}$  and $\nu_i=-2\frac{\lambda_{i1}}{\lambda_{i2}}$. Since the M-SSA solution  is derived from the $nL-2>0$ dimensional intersection of the ellipse defined in Equation \eqref{unit_circle_mult} and  the HT hyperbola characterized in Equation \eqref{hyperbola_mult}, we can infer that the  objective function is not overruled by the constraint, implying that $|\lambda_{i2}|<\infty$. Consequently, $D_i\neq 0$ in Equation \eqref{diff_non_hom_matrix_mult} holds true. The eigenvalues of $\mathbf{N}_i$ are given by $\tilde{\sigma}_j(2\lambda_{k}-\nu_i)$, $1\leq k\leq L$, $1\leq j\leq n$. It is important to note that if $\mathbf{b}_i$ is a solution to the M-SSA problem, then $\nu_i/2$ cannot be an eigenvalue of $\mathbf{M}$. To illustrate this, assume $\nu_i/2=\lambda_{k_0}$ for some $k_0\in\{1,...,L\}$. Under this assumption, $\mathbf{N}_i$ would   map the eigenvectors $\tilde{\mathbf{v}}_{k_0j}$, for $j=1,...,n$,  to zero in Equation \eqref{diff_non_hom_matrix_mult}. We then look at the left-hand side of the equation:   
\[
D_i\tilde{\mathbf{I}}\boldsymbol{\gamma}_{i\cdot\delta}=D_i\tilde{\mathbf{I}}\sum_{k=1}^L\sum_{j=1}^nw_{ikj}\tilde{\mathbf{v}}_{kj}=D_i\sum_{k=1}^L\sum_{j=1}^n\tilde{\sigma}_jw_{ikj}\tilde{\mathbf{v}}_{kj}.
\]
If $\tilde{\mathbf{v}}_{k_0j}$, $j=1,...,n$, is mapped to zero, then $D_i\tilde{\sigma}_jw_{ik_0j}=0$, $j=1,...,n$. Consequently, $w_{ik_0j}=0$, $j=1,...,n$, (because $D_i\tilde{\sigma}_j\neq 0$), thereby contradicting the last regularity assumption (completeness). Thus, we conclude that $\nu_i/2$ cannot be an eigenvalue of $\mathbf{M}$, $\nu_i\in \mathbb{R}-\{2\lambda_k|k=1,...,L\}$, $\mathbf{N}_i^{-1}$ exists and 
\[
\mathbf{N}_i^{-1}=\mathbf{\tilde{V}}\mathbf{D}_{\nu i}^{-1}\mathbf{\tilde{V}}'.
\] 
The diagonal matrix $\mathbf{D}_{\nu i}^{-1}$  is characterized by well-defined entries of the form  $1/(\tilde{\sigma}_j(2\lambda_{k}-\nu_i))$. We can proceed to solve Equation  \eqref{diff_non_hom_matrix_mult} for $\mathbf{b}_i$, yielding the following expression:
\begin{eqnarray}\label{diff_non_hom_matrixe_mult}
\mathbf{b}_i&=&D_i\mathbf{N}_i^{-1}\mathbf{\tilde{I}}\boldsymbol{\gamma}_{i\cdot\delta}\\
&=&D_i\mathbf{\tilde{V}}\mathbf{D}_{\nu i}^{-1}\mathbf{\tilde{V}}' \mathbf{\tilde{I}}\mathbf{\tilde{V}}\mathbf{w}_i\nonumber\\
&=&D_i\sum_{k=1}^L\sum_{j=1}^n \frac{\tilde{\sigma}_jw_{ikj}}{\tilde{\sigma}_j(2\lambda_{k}-\nu_i)}\mathbf{\tilde{v}}_{kj}\nonumber\\
&=&D_i\sum_{k=1}^L\frac{1}{(2\lambda_{k}-\nu_i)}\left(\sum_{j=1}^n w_{ikj}\mathbf{\tilde{v}}_{kj}\right).\label{specdecb_mult0}
\end{eqnarray}
The proof of the second assertion is derived from Equation \eqref{specdecb_mult0} and relies on the orthonormality of $\tilde{\mathbf{v}}_{kj}$: 
\begin{eqnarray}
\rho(\nu_i)&=&\frac{\mathbf{b}_{i}(\nu_i)'\mathbf{\tilde{M}}\mathbf{b}_{i}(\nu_i)}{\mathbf{b}_{i}(\nu_i)'\mathbf{\tilde{I}}\mathbf{b}_{i}(\nu_i)}\\
&=&\displaystyle{\frac{\left(D_i\sum_{k=1}^L\sum_{j=1}^n  \frac{w_{ikj}}{2\lambda_{k }-\nu_i}\mathbf{\tilde{v}}_{kj}\right)'\mathbf{\tilde{M}}\left(D_i\sum_{k=1}^L\sum_{j=1}^n  \frac{w_{ikj}}{2\lambda_{k }-\nu_i}\mathbf{\tilde{v}}_{kj}\right)}{\left(D_i\sum_{k=1}^L\sum_{j=1}^n  \frac{w_{ikj}}{2\lambda_{k }-\nu_i}\mathbf{\tilde{v}}_{kj}\right)'\mathbf{\tilde{I}}\left(D_i\sum_{k=1}^L\sum_{j=1}^n  \frac{w_{ikj}}{2\lambda_{k }-\nu_i}\mathbf{\tilde{v}}_{kj}\right)}}\nonumber\\
&=&\displaystyle{\frac{\sum_{k=1}^{L}\lambda_k\frac{1}{(2\lambda_k-\nu_i)^2}\left(\sum_{j=1}^{n}\tilde{\sigma}_jw_{ikj}^2\right)}{\sum_{k=1}^{L}\frac{1}{(2\lambda_k-\nu_i)^2}\left(\sum_{j=1}^{n}\tilde{\sigma}_jw_{ikj}^2\right)}}.\label{autocorr_mult} 
\end{eqnarray}
It is important to note that the first-order autocorrelation $\rho(\nu_i)$ of $y_{it}$ is influenced by $\boldsymbol{\Sigma}$, since the eigenvalues $\tilde{\sigma}_j$ do not cancel in Equation \eqref{autocorr_mult}. Consequently, selecting $\nu_i = \nu_{0i}$ such that $\rho_i(\nu_{0i}) = \rho_i$ (in accordance with the HT constraint) indicates that the SSA solution $\mathbf{b}_i(\nu_{0i})$ generally depends on $\boldsymbol{\Sigma}$. 
This is despite the fact that, as specified by Equation \eqref{specdecb_mult0}, the expression for  $\mathbf{b}_i$ appears to be independent of $\boldsymbol{\Sigma}$, owing to the cancellation of the eigenvalues $\tilde{\sigma}_j$ in this expression. As $\nu_i$ approaches $2\lambda_k$, we find that $\rho(\nu_i)$ converges to $\lambda_k$, for $1\leq k\leq L$, given the assumption that $\sum_{j=1}^{n}\tilde{\sigma}_jw_{ikj}^2>0$ for $k=1,...,L$. Thus, $\rho(\nu_i)$ can come arbitrarily close to  the boundary values $\pm\rho_{max}(L)$. Therefore, the continuity of $\rho(\nu_i)$ along with the intermediate value theorem implies that any value $\rho_i$ satisfying $|\rho_i|<\rho_{max}(L)$ is permissible in the HT constraint, as asserted.  \\ 
We now proceed to prove Assertion \eqref{ass4_mult}. For this purpose, we define $\tilde{w}_{ik}^2:=\sum_{j=1}^{n}\tilde{\sigma}_jw_{ikj}^2$ in Equation \eqref{autocorr_mult}, resulting in the expression
\begin{eqnarray}
&&\frac{\partial}{\partial\nu_i}\rho(\nu_i)=\frac{\partial}{\partial\nu_i}\displaystyle{\frac{\sum_{k=1}^{L}\lambda_k\frac{1}{(2\lambda_k-\nu_i)^2}\tilde{w}_{ik}^2}{\sum_{k=1}^{L}\frac{1}{(2\lambda_k-\nu_i)^2}\tilde{w}_{ik}^2}}\nonumber\\
&=&2\displaystyle{\frac{\sum_{k=1}^{L}\lambda_k\frac{1}{(2\lambda_k-\nu_i)^3}\tilde{w}_{ik}^2}{\sum_{k=1}^{L}\frac{1}{(2\lambda_k-\nu_i)^2}\tilde{w}_{ik}^2}}-2\displaystyle{\frac{\sum_{k=1}^{L}\lambda_k\frac{1}{(2\lambda_k-\nu_i)^2}\tilde{w}_{ik}^2\sum_{k=1}^{L}\frac{1}{(2\lambda_k-\nu_i)^3}\tilde{w}_{ik}^2}{\left(\sum_{k=1}^{L}\frac{1}{(2\lambda_k-\nu_i)^2}\tilde{w}_{ik}^2\right)^2}}\nonumber\\
&=&2\displaystyle{\frac{\sum_{k=1}^{L}\lambda_k\frac{1}{(2\lambda_k-\nu_i)^3}\tilde{w}_{ik}^2\sum_{j=1}^{L}\frac{1}{(2\lambda_j-\nu_i)^2}\tilde{w}_{ij}^2-\sum_{k=1}^{L}\lambda_k\frac{1}{(2\lambda_k-\nu_i)^2}\tilde{w}_{ik}^2\sum_{j=1}^{L}\frac{1}{(2\lambda_j-\nu_i)^3}\tilde{w}_{ij}^2}{\left(\sum_{k=1}^{L}\frac{1}{(2\lambda_k-\nu_i)^2}\tilde{w}_{ik}^2\right)^2}}\nonumber\\
&=&2\left(\sum_{k=1}^{L}\lambda_k\frac{1}{(2\lambda_k-\nu_i)^3}\tilde{w}_{ik}^2\sum_{j=1}^{L}\frac{1}{(2\lambda_j-\nu_i)^2}\tilde{w}_{ij}^2-\sum_{k=1}^{L}\lambda_k\frac{1}{(2\lambda_k-\nu_i)^2}\tilde{w}_{ik}^2\sum_{j=1}^{L}\frac{1}{(2\lambda_j-\nu_i)^3}\tilde{w}_{ij}^2\right),\label{all_sum}
\end{eqnarray}
utilizing the fact that $\sum_{k=1}^{L} \frac{1}{(2\lambda_k - \nu_i)^2} \tilde{w}_{ik}^2 = \mathbf{b}_i' \mathbf{\tilde{I}} \mathbf{b}_i = 1$ to derive the last equality. For the case when $k=j$, the terms cancel in Equation \eqref{all_sum}. Therefore, we assume $k\neq j$ and consider the expression:
\begin{eqnarray*}
&&\lambda_k\frac{1}{(2\lambda_k-\nu_i)^3}\tilde{w}_{ik}^2\frac{1}{(2\lambda_j-\nu_i)^2}\tilde{w}_{ij}^2-\lambda_k\frac{1}{(2\lambda_k-\nu_i)^2}\tilde{w}_{ik}^2\frac{1}{(2\lambda_j-\nu_i)^3}\tilde{w}_{ij}^2\\
&=&\lambda_k\tilde{w}_{ik}^2\tilde{w}_{ij}^2\frac{1}{(2\lambda_k-\nu_i)^2}\frac{1}{(2\lambda_j-\nu_i)^2}\left(\frac{1}{2\lambda_k-\nu_i}-\frac{1}{2\lambda_j-\nu_i}\right)\\
&=&\lambda_k\tilde{w}_{ik}^2\tilde{w}_{ij}^2\frac{1}{(2\lambda_k-\nu_i)^3}\frac{1}{(2\lambda_j-\nu_i)^3}\left(2\lambda_j-2\lambda_k\right)\\
&=&\tilde{w}_{ik}^2\tilde{w}_{ij}^2\frac{1}{(2\lambda_k-\nu_i)^3}\frac{1}{(2\lambda_j-\nu_i)^3}\left(2\lambda_k\lambda_j-2\lambda_k^2\right)\\
&=&\tilde{w}_{ik}^2\tilde{w}_{ij}^2\frac{1}{(2\lambda_k-\nu_i)^3}\frac{1}{(2\lambda_j-\nu_i)^3}\left(2\lambda_k\lambda_j-2\lambda_j^2\right),
\end{eqnarray*}
where we interchanged $k$ and $j$ in the last equality, owing to symmetry. Assuming $k\neq j$ and considering the addition of the two symmetric terms corresponding to $(k,j)$ and $(j,k)$, we derive the following expression: 
\begin{eqnarray*}
&&\tilde{w}_{ik}^2\tilde{w}_{ij}^2\frac{1}{(2\lambda_k-\nu_i)^3}\frac{1}{(2\lambda_j-\nu_i)^3}2\left(2\lambda_k\lambda_j-\lambda_k^2-\lambda_j^2\right)\\
&=&-2\tilde{w}_{ik}^2\tilde{w}_{ij}^2\frac{1}{(2\lambda_k-\nu_i)^3}\frac{1}{(2\lambda_j-\nu_i)^3}\left(\lambda_k-\lambda_j\right)^2<0,
\end{eqnarray*}
where the strict inequality is justified by the conditions $|\nu_i|>2\rho_{max}(L)\geq 2|\lambda_k|$, $\lambda_k\neq \lambda_j$, and the assumption $\tilde{w}_{ik}^2=\sum_{j=1}^{n}\tilde{\sigma}_jw_{ikj}^2>0$. Consequently, We infer that \[
\frac{\partial}{\partial\nu_i}\rho(\nu_i)<0.
\]
Furthermore, from Equation \eqref{autocorr_mult} we can deduce the limit 
\[
\lim_{|\nu_i|\to \infty}\rho(\nu_i)= \displaystyle{\frac{\sum_{k=1}^{L}\lambda_k\left(\sum_{j=1}^{n}\tilde{\sigma}_jw_{ikj}^2\right)}{\sum_{k=1}^{L}\left(\sum_{j=1}^{n}\tilde{\sigma}_jw_{ikj}^2\right)}}=\displaystyle{\frac{\boldsymbol{\gamma}_{i\cdot\delta}'\mathbf{\tilde{M}}\boldsymbol{\gamma}_{i\cdot\delta}}{\boldsymbol{\gamma}_{i\cdot\delta}'\mathbf{\tilde{I}}\boldsymbol{\gamma}_{i\cdot\delta}}}=\rho_{i,MSE},
\]
which represents  the first-order ACF of the MSE predictor. In conjunction with the identified strictly negative derivative, we conclude that   $\rho_i(\nu_{1i})>\rho_i(\nu_{2i})$ for $\nu_{1i} > 2\rho_{max}(L)$ and $\nu_{2i}< -2\rho_{max}(L)$, as asserted.\\
For a proof of the last Assertion \eqref{ass5} we analyze the expression
\begin{eqnarray*}
\rho(y_i(\nu_i),z_i,\delta)&=&\frac{\mathbf{b}_i'\mathbf{\tilde{I}}\boldsymbol{\gamma}_{i\cdot\delta}}{\sqrt{\mathbf{b}_i'\mathbf{\tilde{I}}\mathbf{b}_i\boldsymbol{\gamma}_{i\cdot\delta}'\mathbf{\tilde{I}}\boldsymbol{\gamma}_{i\cdot\delta}}}=D_i\frac{\boldsymbol{\gamma}_{i\cdot\delta}'\mathbf{\tilde{I}}'\mathbf{N}_i^{-1}\mathbf{\tilde{I}}\boldsymbol{\gamma}_{i\cdot\delta}}{\sqrt{D_i^2\boldsymbol{\gamma}_{i\cdot\delta}'\mathbf{\tilde{I}}'\mathbf{N}_i^{-1}~'\mathbf{N}_i^{-1}\mathbf{\tilde{I}}\boldsymbol{\gamma}_{i\cdot\delta}\boldsymbol{\gamma}_{i\cdot\delta}'\mathbf{\tilde{I}}\boldsymbol{\gamma}_{i\cdot\delta}}}\\
&=&
D_i\frac{\left(\sum_{k=1}^L\frac{1}{2\lambda_{k}-\nu_i}\sum_{j=1}^n w_{ikj}\mathbf{\tilde{v}}_{kj}'\right)\left(\sum_{k=1}^L\sum_{j=1}^n \tilde{\sigma}_jw_{ikj}\mathbf{\tilde{v}}_{kj}\right)}{\sqrt{D_i^2\boldsymbol{\gamma}_{i\cdot\delta}'\mathbf{\tilde{I}}'\mathbf{N}^{-2}\mathbf{\tilde{I}}\boldsymbol{\gamma}_{i\cdot\delta}\boldsymbol{\gamma}_{i\cdot\delta}'\mathbf{\tilde{I}}\boldsymbol{\gamma}_{i\cdot\delta}}}\\
&=&\textrm{sign}(D_i)\frac{\sum_{k=1}^L\frac{1}{2\lambda_{k}-\nu_i}\sum_{j=1}^n \tilde{\sigma}_jw_{ikj}^2}
{\sqrt{\boldsymbol{\gamma}_{i\cdot\delta}'\mathbf{\tilde{I}}'\mathbf{N}^{-2}\mathbf{\tilde{I}}\boldsymbol{\gamma}_{i\cdot\delta}\boldsymbol{\gamma}_{i\cdot\delta}'\mathbf{\tilde{I}}\boldsymbol{\gamma}_{i\cdot\delta}}},
\end{eqnarray*}
where we used $\mathbf{N}_i^{-1}~'\mathbf{N}_i^{-1}=\mathbf{N}_i^{-1}\mathbf{N}_i^{-1}=:\mathbf{N}^{-2}$, owing to symmetry. 
For the case where $\nu<-2\rho_{max}(L)$, the quotient is strictly positive. Consequently, the positivity of the objective function for the M-SSA solution implies that $\textrm{sign}(D_i)=-\textrm{sign}(\nu_i)=1$. Conversely,  for $\nu>2\rho_{max}(L)$ the quotient becomes strictly negative, leading to $\textrm{sign}(D_i)=-\textrm{sign}(\nu_i)=-1$. Assuming $\nu<-2\rho_{max}(L)$ we deduce:
\begin{eqnarray*}
\frac{d\rho\Big(y_i(\nu_i),z_i,\delta\Big)}{d\nu_i}&=&\frac{d}{d\nu_i}\left(\frac{\boldsymbol{\gamma}_{i\cdot\delta}'\mathbf{\tilde{I}}'\mathbf{N}_i^{-1}\mathbf{\tilde{I}}\boldsymbol{\gamma}_{i\cdot\delta}}{\sqrt{\boldsymbol{\gamma}_{i\cdot\delta}'\mathbf{\tilde{I}}'\mathbf{N}_i^{-2}\mathbf{\tilde{I}}\boldsymbol{\gamma}_{i\cdot\delta}\boldsymbol{\gamma}_{i\cdot\delta}'\mathbf{\tilde{I}}^2\boldsymbol{\gamma}_{i\cdot\delta}}}\right)\\
&=&\frac{1}{\left(\boldsymbol{\gamma}_{i\cdot\delta}'\mathbf{\tilde{I}}'\mathbf{N}_i^{-2}\mathbf{\tilde{I}}\boldsymbol{\gamma}_{i\cdot\delta}\right)^{3/2}\sqrt{\boldsymbol{\gamma}_{i\cdot\delta}'\mathbf{\tilde{I}}^2\boldsymbol{\gamma}_{i\cdot\delta}}} \left\{\left(\boldsymbol{\gamma}_{i\cdot\delta}'\mathbf{\tilde{I}}'\mathbf{N}_i^{-2}\mathbf{\tilde{I}}\boldsymbol{\gamma}_{i\cdot\delta}\right)^2-\boldsymbol{\gamma}_{i\cdot\delta}'\mathbf{\tilde{I}}'\mathbf{N}_i^{-1}\mathbf{\tilde{I}}\boldsymbol{\gamma}_{i\cdot\delta}\boldsymbol{\gamma}_{i\cdot\delta}'\mathbf{\tilde{I}}'\mathbf{N}_i^{-3}\mathbf{\tilde{I}}\boldsymbol{\gamma}_{i\cdot\delta}\right\}\\
&=&-\textrm{sign}(\nu)\frac{\sqrt{\boldsymbol{\gamma}_{i\cdot\delta}'\mathbf{\tilde{I}}'\mathbf{N}_i^{-2}\mathbf{\tilde{I}}\boldsymbol{\gamma}_{i\cdot\delta}}}{\sqrt{\boldsymbol{\gamma}_{i\cdot\delta}'\mathbf{\tilde{I}}^2\boldsymbol{\gamma}_{i\cdot\delta}}}\frac{d\rho(\nu_i)}{d\nu_i}<0,
\end{eqnarray*}
where we inserted $-\textrm{sign}(\nu)=1$. The last expression is obtained by recognizing that
\begin{eqnarray}
\frac{d\rho(\nu_i)}{d\nu_i}&=&\frac{d}{d\nu_i}\left(\frac{\mathbf{b}_i'\mathbf{\tilde{M}b}_i}{\mathbf{b}_i'\mathbf{\tilde{I}}\mathbf{b}_i}\right)=\frac{d}{d\nu_i}\left(\frac{\boldsymbol{\gamma}_{i\cdot\delta}'\mathbf{\tilde{I}}'\mathbf{N}_i^{-1}~'\mathbf{\tilde{M}}\mathbf{N}_i^{-1}\mathbf{\tilde{I}}\boldsymbol{\gamma}_{i\cdot\delta}}{\boldsymbol{\gamma}_{i\cdot\delta}'\mathbf{\tilde{I}}'\mathbf{N}_i^{-1}~'\mathbf{N}_i^{-1}\mathbf{\tilde{I}}\boldsymbol{\gamma}_{i\cdot\delta}}\right)=\frac{d}{d\nu_i}\left(\frac{\boldsymbol{\gamma}_{i\cdot\delta}'\mathbf{\tilde{I}}'\mathbf{\tilde{M}}\mathbf{N}_i^{-2}\mathbf{\tilde{I}}\boldsymbol{\gamma}_{i\cdot\delta}}{\boldsymbol{\gamma}_{i\cdot\delta}'\mathbf{\tilde{I}}'\mathbf{N}_i^{-2}\mathbf{\tilde{I}}\boldsymbol{\gamma}_{i\cdot\delta}}\right)\nonumber\\
&=&\frac{2\boldsymbol{\gamma}_{i\cdot\delta}'\mathbf{\tilde{I}}'\mathbf{\tilde{M}}\mathbf{N}_i^{-3}\mathbf{\tilde{I}}\boldsymbol{\gamma}_{i\cdot\delta}\cdot\boldsymbol{\gamma}_{i\cdot\delta}'\mathbf{\tilde{I}}'\mathbf{N}_i^{-2}\mathbf{\tilde{I}}\boldsymbol{\gamma}_{i\cdot\delta}
-2\boldsymbol{\gamma}_{i\cdot\delta}'\mathbf{\tilde{I}}'\mathbf{\tilde{M}}\mathbf{N}_i^{-2}\mathbf{\tilde{I}}\boldsymbol{\gamma}_{i\cdot\delta}\cdot\boldsymbol{\gamma}_{i\cdot\delta}'\mathbf{\tilde{I}}'\mathbf{N}_i^{-3}\mathbf{\tilde{I}}\boldsymbol{\gamma}_{i\cdot\delta}}{(\boldsymbol{\gamma}_{i\cdot\delta}'\mathbf{\tilde{I}}'\mathbf{N}_i^{-2}\mathbf{\tilde{I}}\boldsymbol{\gamma}_{i\cdot\delta})^2}\label{eq_diff_1}\\
&=&\frac{(\boldsymbol{\gamma}_{i\cdot\delta}'\mathbf{\tilde{I}}'\mathbf{N}_i^{-2}\mathbf{\tilde{I}}\boldsymbol{\gamma}_{i\cdot\delta})^2-\boldsymbol{\gamma}_{i\cdot\delta}'\mathbf{\tilde{I}}'\mathbf{N}_i^{-1}\mathbf{\tilde{I}}\boldsymbol{\gamma}_{i\cdot\delta}\cdot\boldsymbol{\gamma}_{i\cdot\delta}'\mathbf{\tilde{I}}'\mathbf{N}_i^{-3}\mathbf{\tilde{I}}\boldsymbol{\gamma}_{i\cdot\delta}}{(\boldsymbol{\gamma}_{i\cdot\delta}'\mathbf{\tilde{I}}'\mathbf{N}_i^{-2}\mathbf{\tilde{I}}\boldsymbol{\gamma}_{i\cdot\delta})^2},\label{eq_diff_2}
\end{eqnarray}
where $\mathbf{N}_i^{-k}:=(\mathbf{N}_i^{-1})^k$, $(\mathbf{N}_i^{-1})'=\mathbf{N}_i^{-1}$ (by symmetry). The commutativity of the matrix product utilized in deriving the third equality arises from the simultaneous diagonalization of both $\mathbf{\tilde{M}}$ and $\mathbf{N}_i^{-1}$, which share the same eigenvectors. Additionally, standard rules for matrix differentiation were employed in the derivation of Equation \eqref{eq_diff_1}\footnote{$\frac{d(\mathbf{N}_i^{-1})}{d\nu_i}=\mathbf{N}_i^{-2}$ and $\frac{d(\mathbf{N}_i^{-2})}{d\nu_i}=2\mathbf{N}_i^{-3}$. The first equation follows from the general differentiation rule  $\frac{d(\mathbf{N}_i^{-1})}{d\nu_i}=-\mathbf{N}_i^{-1}\frac{d\mathbf{N}_i}{d\nu_i}\mathbf{N}_i^{-1}$, noting that $\frac{d\mathbf{N}_i}{d\nu_i}=-\mathbf{\tilde{I}}$. The second equation is derived by substituting the first equation into  $\frac{d(\mathbf{N}_i^{-2})}{d\nu_i}=\frac{d(\mathbf{N}_i^{-1})}{d\nu_i}\mathbf{N}_i^{-1}+\mathbf{N}_i^{-1}\frac{d(\mathbf{N}_i^{-1})}{d\nu_i}$.}. Finally, the expression $2\mathbf{\tilde{M}}\mathbf{N}_i^{-k}=
\mathbf{N}_i^{-k+1}+\nu\mathbf{N}_i^{-k}$ was incorporated into the numerator of Equation \eqref{eq_diff_1}, leading to the final equation after simplification.  
The aforementioned proof is also applicable in the scenario where $\nu > 2\rho_{\max}(L)$, albeit with a change in sign, such that $\text{sign}(D_i) = -\text{sign}(\nu) = -1$, as was to be demonstrated. \hfill \qedsymbol{}\\

\textbf{Remarks}\\
As $|\nu_i|\to \infty$ and $D_i=-\nu_i$, it follows that $-\mathbf{N}_i/\nu_i\to \mathbf{\tilde{I}}$. Consequently, we have
\[
\mathbf{b}_i\to-\frac{D_i}{\nu_i}\tilde{\mathbf{I}}^{-1} \tilde{\mathbf{I}}\boldsymbol{\gamma}_{i\cdot\delta}=\boldsymbol{\gamma}_{i\cdot\delta}.
\]
This indicates that the M-SSA method asymptotically replicates the MSE predictor in the limit as $|\nu_i| \to \infty$, reflecting an asymptotically degenerate case where the HT constraint may be omitted. By Assertion \eqref{ass4_mult}, we conclude that M-SSA can accommodate HT constraints $\rho_i<\rho_{i,MSE}$ when $\nu_i\in]-\infty,-2\rho_{max}(L)]$. Conversely, for $\nu_i\in[2\rho_{max}(L),\infty[$ M-SSA is capable of accommodating HT constraints $\rho_i>\rho_{i,MSE}$ (smoothing). Moreover, Equation \eqref{ficcc} elucidates a trade-off between Accuracy (target correlation) and Smoothness (first-order ACF) in the context of solving the M-SSA optimization problem. For a more comprehensive discussion, refer to Wildi (2024). Finally, Wildi (2026) also addresses the singular case of incomplete spectral support, which may be generalized to M-SSA; however, this topic is omitted here for the sake of brevity (and lack of practical relevance). \\

\begin{Corollary}\label{block_system} 
Let the assumptions of Theorem \eqref{lambda_mult} be satisfied. 
\begin{enumerate}
\item For each $i=1,...,n$, the solution $\mathbf{b}_i=(\mathbf{b}_{i1}',\mathbf{b}_{i2}',...,\mathbf{b}_{in}')'$ to the M-SSA optimization problem can alternatively be derived from the $n$ subsystems given by   
\begin{eqnarray}\label{ssa_mult_break}
\mathbf{b}_{ij}=D_i\boldsymbol{{\nu}}_i^{-1}\boldsymbol{\gamma}_{ij\delta},
\end{eqnarray}
for $j=1,...,n$, where $D_i\neq 0$ and $\boldsymbol{{\nu}}_i:=2\mathbf{{M}}-\nu_i\mathbf{{I}}$ represents an invertible $L\times L$ matrix. The  subsystems are interconnected through the common parameters $D_i$ and $\nu_i$, while they generally differ with respect to $\boldsymbol{\gamma}_{ij\delta}$, $j=1,...,n$.
\item \label{cor1ass}For each $i=1,...,n$, the M-SSA predictor satisfies the following non-stationary and time reversible difference equation:
\begin{eqnarray}\label{time_domain_mult}
b_{ijk+1}-\nu_i b_{ijk}+b_{ijk-1}&=&D_i\gamma_{ijk+\delta}~,~0\leq k\leq L-1,~j=1,...,n,~i=1,...,n
\end{eqnarray}
with implicit boundary conditions $b_{ij,-1}=b_{ijL}=0$ to ensure stability of the solution. 
\end{enumerate}
\end{Corollary}

\textbf{Proof}: The first claim can be substantiated by multiplying both sides of Equation \eqref{diff_non_hom_matrix_mult} by $\mathbf{\tilde{I}}^{-1}=\boldsymbol{\Sigma}^{-1}\otimes \mathbf{I}_{L\times L}$. This yields the following equations:
\begin{eqnarray*}
D_i\boldsymbol{\gamma}_{i\cdot\delta}&=& \mathbf{\tilde{I}}^{-1}\mathbf{N}_i\mathbf{b}_i\\
&=&\Big(2(\mathbf{I}_{n\times n}\otimes\mathbf{{M}})-\nu_i\mathbf{{I}}_{nL\times nL}\Big)\mathbf{b}_i.
\end{eqnarray*}
A proof follows from the block-diagonal structure of $\mathbf{I}_{n\times n}\otimes\mathbf{{M}}$. The validity of the second claim follows directly from Equation \eqref{ssa_mult_break}, achieved by multiplying both sides by $\boldsymbol{{\nu}}_i$. Specifically, the boundary equations $b_{ij1}-\nu_i b_{ij0}=D_i\gamma_{ij\delta}$ and $-\nu_i b_{ij(L-1)}+b_{ij(L-2)}=D_i\gamma_{ij(L-1+\delta)}$ at $k=0$ and $k=L-1$ substantiate the implicit boundary conditions $b_{ij,-1}=0$ and $b_{ijL}=0$.\\

Theorem \eqref{lambda_mult} presents a one-parameter formulation of the SSA solution, and the subsequent corollary elucidates this solution by establishing a relationship between the unknown parameter $\nu_i$ and the HT constraint.

\begin{Corollary}\label{lambda_num_gen_mult}
Let the assumptions of theorem \eqref{lambda_mult} be satisfied. Then, for each $i=1,...,n$, the solution to the SSA-optimization problem described in Equation \eqref{mcrit1} is given by $s_i\mathbf{b}_i(\nu_{0i})$, where $\mathbf{b}_i(\nu_{0i})$ is derived from Equation \eqref{ssa_mult}. This is under the assumption of an arbitrary scaling such that $|D_i|=1$ (with the sign of $D_i$ determined by the requirement for a positive objective function). Here, $\nu_{0i}$ represents a solution to the non-linear HT equation $\rho(\nu_i)=\rho_i$, and the scaling $s_i$ is defined as $s_i=\sqrt{l/\mathbf{b}_i(\nu_{0i})'\tilde{\mathbf{I}}\mathbf{b}_i(\nu_{0i})}$ to meet the length constraint. If the search for an optimal $\nu_{0i}$ can be confined to the domain $\{\nu||\nu|>2\rho_{max}(L)\}$, then $\nu_{0i}$ is uniquely determined by $\rho_i$. 
\end{Corollary}

A proof follows directly from Theorem \eqref{lambda_mult}, with the observation that the scaling factor $s_i = \sqrt{l / \mathbf{b}_i(\nu_{0i})' \tilde{\mathbf{I}} \mathbf{b}_i(\nu_{0i})}$ does not interfere with either the objective function or the HT constraint. This scaling can be applied subsequently, after obtaining a solution under the assumption of an arbitrary scaling such that $|D_i| = 1$. Notably, assertion \eqref{ass4_mult} of the theorem guarantees the uniqueness of solutions within the set $\{\nu | |\nu| > 2 \rho_{\text{max}}(L)\}$, owing to the bijectivity of the first order ACF. \hfill \qedsymbol{}\\

\textbf{Remark}: The objective function $\rho(y_i,{z}_{i\delta},\delta)$ and the first-order ACF $\rho(y_{i})$ are generally influenced by  $\boldsymbol{\Sigma}$. Interestingly, Equation \eqref{ssa_mult_break} indicates that the one-parameter form of the solution $\mathbf{b}_i(\nu_i)$ is independent of $\boldsymbol{\Sigma}$; this is corroborated by the cancellation of $\tilde{\sigma}_j$ in Equation \eqref{specdecb_mult0}. However, the optimal solution $\mathbf{b}_i(\nu_{0i})$, which is derived from the optimal $\nu_{0i}$, remains dependent on $\boldsymbol{\Sigma}$ because $\tilde{\sigma}_j$ (for $j = 1, \ldots, n$) does not cancel in the first-order ACF \eqref{spec_dec_rho_mult}, which is critical for determining $\nu_{i0}$ within the context of the nonlinear HT constraint.\\

The next result addresses the distribution of the M-SSA predictor.  
\begin{Corollary}
Let all regularity assumptions of Theorem \eqref{lambda_mult} be satisified, and let $\hat{\boldsymbol{\gamma}}_{i\cdot\delta}$ represent a finite-sample estimate of the MSE-predictor ${\boldsymbol{\gamma}}_{i\cdot\delta}$, with mean ${\boldsymbol{\mu}}_{\gamma_{i\delta}}$ and variance ${\boldsymbol{\Sigma}}_{\gamma_{i\delta}}$. Then, the mean and variance of the M-SSA predictor $\hat{\mathbf{b}}_i$ are expressed as follows:
\begin{eqnarray*}
{\boldsymbol{\mu}}_{\mathbf{b}_i}&=&D_i\mathbf{N}_i^{-1}\mathbf{\tilde{I}}{\boldsymbol{\mu}}_{\gamma_{i\delta}}\\
{\boldsymbol{\Sigma}}_{\mathbf{b}_i}&=&D_i^2\mathbf{N}_i^{-1}\mathbf{\tilde{I}}{\boldsymbol{\Sigma}}_{\gamma_{i\delta}}\mathbf{\tilde{I}}\mathbf{N}_i^{-1}
\end{eqnarray*}
If $\hat{\boldsymbol{\gamma}}_{i\cdot\delta}$ follows a Gaussian distribution, then $\hat{\mathbf{b}}_i$ will also be Gaussian distributed. 
\end{Corollary}
The proof follows directly from Equation \eqref{ssa_mult}, acknowledging the symmetry of $\mathbf{N}_i^{-1}$. For a comprehensive derivation of the mean, variance, and (asymptotic) distribution of the MSE estimate, one may consult standard textbooks, such as Brockwell and Davis (1993).\\

The final result in this section introduces a dual reformulation of the M-SSA optimization criterion, which characterizes its solution as the smoothest predictor among all predictors that achieve the same tracking accuracy or target correlation. To facilitate this, we introduce some notation: 
\begin{equation}\label{def_mse_rhomax}
\mathbf{y}_{it,MSE}^{\rho_{max}}:=\boldsymbol{\gamma}_{i\cdot\delta}^{\rho_{max}}~'\mathbf{x}_{\cdot t}
\end{equation}
is the output of the filter with weights $\boldsymbol{\gamma}_{i\cdot\delta}^{\rho_{max}}:=\sum_{j=1}^n w_{i1j}\mathbf{\tilde{v}}_{1j}$, corresponding to the projection of the MSE predictor onto the subspace spanned by $\mathbf{\tilde{v}}_{11},...,\mathbf{\tilde{v}}_{1n}$. Under the regularity assumptions of Theorem \eqref{lambda_mult}, specifically a complete spectral support, $\boldsymbol{\gamma}_{i\cdot\delta}^{\rho_{max}}\neq\mathbf{0}$.   According to the proof of the second assertion in Proposition \eqref{bound12}, $\mathbf{y}_{it,MSE}^{\rho_{max}}$ maximizes the target correlation under the (extremal) HT constraint $\rho(\mathbf{y}_{it,MSE}^{\rho_{max}})=\rho_{max}(L)$ (boundary point). Consequently, if $y_{it}$ is such that $\rho(y_i,z_i,\delta)>\rho({y}_{i,MSE}^{\rho_{max}},z_i,\delta)$ for its target correlation, then $\rho(y_i)<\rho_{max}(L)$ for its first-order ACF (interior point). Moreover, under the assumption of complete spectral support, $\boldsymbol{\gamma}_{i\cdot\delta}^{\rho_{max}}\neq \boldsymbol{\gamma}_{i\cdot\delta}$ (since the latter's spectral weights $w_{ikj}$ cannot all vanish for $k>1$)  and therefore the strict inequality  $\rho({y}_{i,MSE}^{\rho_{max}},z_i,\delta)<\rho(y_{i,MSE},z_i,\delta)$ holds for the target correlation of the MSE predictor, with filter weights $\boldsymbol{\gamma}_{i\cdot\delta}$, owing to uniqueness of the MSE predictor.\\

\begin{Theorem}\label{cor3}
Consider the dual optimization problem expressed as 
\begin{eqnarray}\label{crit2}
\left.\begin{array}{cc}
&\max_{\mathbf{b}_i}\rho(y_{i})\\
&\rho(y_i,z_i,\delta)=\rho_{iyz}\\
&\mathbf{b}_i'\mathbf{\tilde{I}b}_i=l
\end{array}\right\}, i=1,...,n,
\end{eqnarray}
where the roles of the first-order autocorrelation—now incorporated into the objective function—and the target correlation—now specified as a constraint—are interchanged. Let the assumptions outlined in Theorem \eqref{lambda_mult} hold, except that the first regularity condition (non-degeneration) is replaced by $\rho_{iyz}>\rho({y}_{i,MSE}^{\rho_{max}},z_i,\delta)$ (the target correlation of $\mathbf{y}_{it,MSE}^{\rho_{max}}$ defined in Equation \eqref{def_mse_rhomax}) and the second regularity condition (interior point) is replaced by $|\rho_{iyz}|<\rho(y_{i,MSE},z_i,\delta)$ (the target correlation of the $i$-th MSE predictor). Then the following results hold:
\begin{enumerate}
\item The solution to the dual problem adopts the same parametric form as the original M-SSA solution:
\begin{equation}\label{doc}
\mathbf{{b}}_{i}=\tilde{D}_i\mathbf{\tilde{N}}_i^{-1}\mathbf{\tilde{I}}\boldsymbol{\gamma}_{i\cdot\delta},
\end{equation}
where $\mathbf{\tilde{N}}:=(2\tilde{\mathbf{M}}-\tilde{\nu}_i\mathbf{\tilde{I}})$ is a full-rank matrix and $\tilde{D}_i\neq 0,\tilde{\nu}_i$ can be chosen to satisfy the constraints. 
\item Let $y_{it}(\nu_{0i})$ represent the original M-SSA solution for the $i$-th target, assuming $\nu_{0i}>2\rho_{max}(L)$, and set $\rho_{iyz}:=\rho(y_i(\nu_{0i}),z_i,\delta)$  (the value of the $i$-th maximized M-SSA objective function) in the target correlation constraint of the dual criterion \eqref{crit2}. If the search for an optimal  $\tilde{\nu}_i$ in the specified dual problem can be confined to the domain $\{\nu||\nu|>2\rho_{max}(L)\}$, then the solution $y_{it}(\nu_{0i})$ of the original M-SSA problem also constitutes an optimal solution to the dual problem.
\item Conversely, if $\nu_{0i}<-2\rho_{max}(L)$, then the solution $y_{it}(\nu_{0i})$ remains optimal for the dual problem, provided that the objective of Criterion \eqref{crit2} is reformulated from a maximization to a minimization problem.
\end{enumerate}
\end{Theorem}

\textbf{Proof}: 
Given the similarity in the problem structure, we may leverage the arguments employed in the proof of Theorem \eqref{lambda_mult} to verify the above claims. Accordingly, we here focus the analysis on the relevant deviations. The Lagrangian $\tilde{\mathcal{L}}_i$ of the dual problem is expressed as:
\begin{eqnarray*}
\tilde{\mathcal{L}}_i:=\mathbf{b}_{i}'\mathbf{\tilde{M}}\mathbf{b}_{i}-
\tilde{\lambda}_{i1}(\mathbf{b}_{i}'\mathbf{\tilde{I}}\mathbf{b}_{i}-l)-\tilde{\lambda}_{i2}\left(\frac{\boldsymbol{\gamma}_{i\cdot\delta}'\tilde{\mathbf{I}}\mathbf{b}_{i}}{\sqrt{l\boldsymbol{\gamma}_{i\cdot\delta}'\mathbf{\tilde{I}}\boldsymbol{\gamma}_{i\cdot\delta}}}-\rho_{iyz}\right).
\end{eqnarray*}
where we inserted $\rho(y_i,z_i,\delta)=\frac{\boldsymbol{\gamma}_{i\cdot\delta}'\tilde{\mathbf{I}}\mathbf{b}_{i}}{\sqrt{l\boldsymbol{\gamma}_{i\cdot\delta}'\mathbf{\tilde{I}}\boldsymbol{\gamma}_{i\cdot\delta}}}$ in the target correlation constraint. 
The two new regularity conditions ensure that the solution to the dual criterion is neither an unconditional maximum of the new objective function (non-degeneration) nor the maximum in the target correlation of the new constraint (interior point), see the reasoning following Eq.\eqref{def_mse_rhomax}. As a result, Lagrange multipliers are finite and non-vanishing. After differentiation of the Lagrangian $\tilde{\mathcal{L}}_i$,  Equation \eqref{doc} is obtained, after re-arranging terms\footnote{The dependence of $\tilde{D}_i,\tilde{\nu}_i$ on the Lagrangian multipliers $\tilde{\lambda}_{i1},\tilde{\lambda}_{i2}$ differs slightly from that in the primal problem, a discrepancy that can be rectified through straightforward re-scaling, given that the multipliers are finite and non-vanishing.}. The feasible solution space is now defined by the intersection of the length-ellipse—corresponding to the original length constraint—and the target correlation plane—representing the new target correlation constraint (in contrast to the intersection of the ellipse and hyperbola in the primal problem).\\
For the proof of the second assertion, we first note that the main distinction between original M-SSA and dual criteria concerns  the selection of $\tilde{\nu}_i$, such that it satisfies the new target correlation constraint $\rho(y_i(\tilde{\nu_i}),z_i,\delta)=\rho(y_i(\nu_{0i}),z_i,\delta)$ (instead of the original HT constraint). This distinction is significant insofar as the optimal $\tilde{\nu}_{i0}$ of the dual problem is generally not uniquely specified by the hyperparameter $\rho_{iyz}$ associated with the target correlation constraint\footnote{Unlike the first-order ACF, the target correlation is not a strictly monotonic function for  $\nu_i\in\{\nu||\nu|>2\rho_{max}(L)\}$, since the sign of its derivative depends on the sign of $\nu_i$, see Equation \eqref{ficcc}.}, in contrast to the primal (original M-SSA) problem, where the parameter $\nu_{0i}$ is uniquely determined by $\rho_i$ in the HT constraint (at least under the posited assumptions, see Corollary \eqref{lambda_num_gen_mult}). However, we infer from the similar formal structure of the problem, as expressed by Equation \eqref{doc}, that the strict monotonicity of the 
target correlation $\rho(y_i(\tilde{\nu_i}),z_i,\delta)$, when  $\tilde{\nu}_i\in\{\nu|\nu>2\rho_{max}(L)\}$ or when $\tilde{\nu}_i\in\{\nu|\nu<-2\rho_{max}(L)\}$ (Assertion \eqref{ass5} of Theorem \eqref{lambda_mult}), extends to the dual problem. To proceed with the proof of the second claim, we then initially assume that the search space for  $\tilde{\nu}_i$ in the dual problem can be restricted to the set $\{\nu|\nu>2\rho_{max}(L)\}$. In this case, the solutions to the primal and dual problems must coincide due to strict monotonicity of $\rho(y_i(\tilde{\nu_i}),z_i,\delta)$, which guarantees the uniqueness of the solution. 
The requested extension of this result to the set $\{\nu||\nu|>2\rho_{max}(L)\}$ is further supported by Assertion \eqref{ass4_mult} of Theorem \eqref{lambda_mult}, which applies in analogous manner to the dual problem and which establishes that 
$\rho(\tilde{\nu}_i)<\rho(\nu_{0i})$, when $\tilde{\nu}_i<-2\rho_{max}(L)$ (since $\nu_{0i}>2\rho_{max}(L)$, by assumption). Consequently, the maximum of the objective function of the dual criterion can be confined to the region $\{\nu|\nu>2\rho_{max}(L)\}$. \\
A similar line of reasoning applies to the proof of the last claim, where $\nu_{0i}<-2\rho_{max}(L)$ in the primal problem, with the necessary adjustment that the maximization in the dual Criterion \eqref{crit2} must be replaced with minimization, as now $\rho(\tilde{\nu}_i)>\rho(\nu_{i0})$ if $\tilde{\nu}_i>2\rho_{max}(L)$\footnote{If the maximization were retained in the objective function of the dual criterion, it would generally result in swapping the solution $\nu_{0i}<-2\rho_{max}(L)$ of the primal problem from the left branch $\{\nu|\nu<-2\rho_{max}(L)\}$ to the right branch $\{\nu|\nu>2\rho_{max}(L)\}$, due to maximization of the first-order ACF. Such a swap would lead to a different solution to the dual problem in this case.}.\hfill \qedsymbol{}\\

\textbf{Remark:} Wildi (2026) demonstrates that the condition  $|\nu_i|>2\rho_{max}(L)$, ensuring strict monotonicity and uniqueness in the previous results, is not a limitation in practical applications.  Consequently, we may assume this condition holds except in highly unusual problem specifications.\\

Under the posited assumptions, the above theorem characterizes the M-SSA predictor as the smoothest predictor—exhibiting the lowest zero-crossing rate—for a given level of tracking accuracy (target correlation). This property distinguishes the M-SSA approach as a compelling alternative to traditional smoothing methods, as demonstrated in the next section.

\section{Applications}\label{examples}

Our examples highlight three key aspects: smoothing, forecasting, and nowcasting (real-time signal extraction). Smoothing is primarily concerned with causal targets, whereas traditional forecasting typically focuses on acausal, forward-looking allpass filters. Signal extraction encompasses various acausal filtering techniques, including lowpass filters for trend extraction, bandpass filters for cycle analysis, highpass filters for separating cycles from noise, and seasonal adjustment.

\subsection{Forecasting}\label{var1}

We apply M-SSA to a stationary VAR(1) process represented by the equation:
\[
\left(\begin{array}{c}z_{1t}\\z_{2t}\end{array}\right)=\mathbf{A}_1\left(\begin{array}{c}z_{1t-1}\\z_{2t-1}\end{array}\right)+\left(\begin{array}{c}\epsilon_{1t}\\\epsilon_{2t}\end{array}\right)~~~,~~~\mathbf{A}_1=\left(\begin{array}{cc}0.7&0.4\\-0.6&0.9\end{array}\right)~~,~~\boldsymbol{\Sigma}=\left(\begin{array}{cc}1.09&-1.45\\-1.45&2.58\end{array}\right).
\]
We compute one-step ahead  predictors  for $z_{1t+1}$ and $z_{2t+1}$ with $\delta=1$ and a predictor length of $L=100$. In the classic forecast application, $\mathbf{z}_t=\mathbf{x}_t$, $\boldsymbol{\Gamma}_k=\left\{\begin{array}{cc}\mathbf{I}&k=0\\\mathbf{0}&\textrm{otherwise}\end{array}\right.$ and the target $\mathbf{z}_{t+\delta}$ relies on an anticipative all-pass filter, for $\delta>0$. 
The HTs within the constraint are defined as   $ht_1=3$, $ht_2=8$. 
The VAR process allows for a convergent MA-inversion with weights $\boldsymbol{\Xi}_k=\mathbf{A}_1^k$. Accordingly, the HTs of the MSE predictor can be computed by substituting the first-order ACFs $(\boldsymbol{\gamma}\cdot\boldsymbol{\xi})_{i\delta}'\mathbf{\tilde{M}}(\boldsymbol{\gamma}\cdot\boldsymbol{\xi})_{i\delta}/(\boldsymbol{\gamma}\cdot\boldsymbol{\xi})_{i\delta}'\mathbf{\tilde{I}}(\boldsymbol{\gamma}\cdot\boldsymbol{\xi})_{i\delta}$, $i=1,2$, $\delta=1$, into Equation \eqref{ht} to yield $ht_{1MSE}=5.6$ and $ht_{2MSE}=4.6$. Here, $(\boldsymbol{\gamma}\cdot\boldsymbol{\xi})_{i\delta}$ represents the convolution of the target and the MA-inversion, as discussed in Section \eqref{ext_stat}. Given that the target is a forward-shifting allpass filter, we conclude that  $(\boldsymbol{\gamma}\cdot\boldsymbol{\xi})_{i\delta}$ corresponds to the shifted original MA-inversion, with the leading term at $-\delta=-1$ omitted (the forecast of the future noise term $\epsilon_{t+1}$ vanishes). \\

With our specifications we observe $ht_1<ht_{1MSE}$ , indicating that the M-SSA predictor must introduce additional zero crossings to `unsmooth' the MSE benchmark for the first series. Conversely, for the second series,  $ht_2>ht_{2MSE}$ suggests that the M-SSA predictor must `smooth' the MSE benchmark by producing fewer zero crossings. 
M-SSA and MSE predictors are displayed in Fig.\eqref{filt_coef_var1}, with the M-SSA predictor presented in the top panel and the MSE predictor in the bottom panel. The weights associated with the MSE predictor are defined as $\mathbf{A}_1^{\delta}=\mathbf{A}_1$. The unsmoothing effect   is characterized by an alternating high frequency pattern exhibited by the M-SSA filter in the first series (top left panel). In contrast, the smoothing effect results in slowly monotonically decaying patterns that are observed in the second series (top right panel).       
\begin{figure}[H]\begin{center}\includegraphics[height=3in, width=6in]{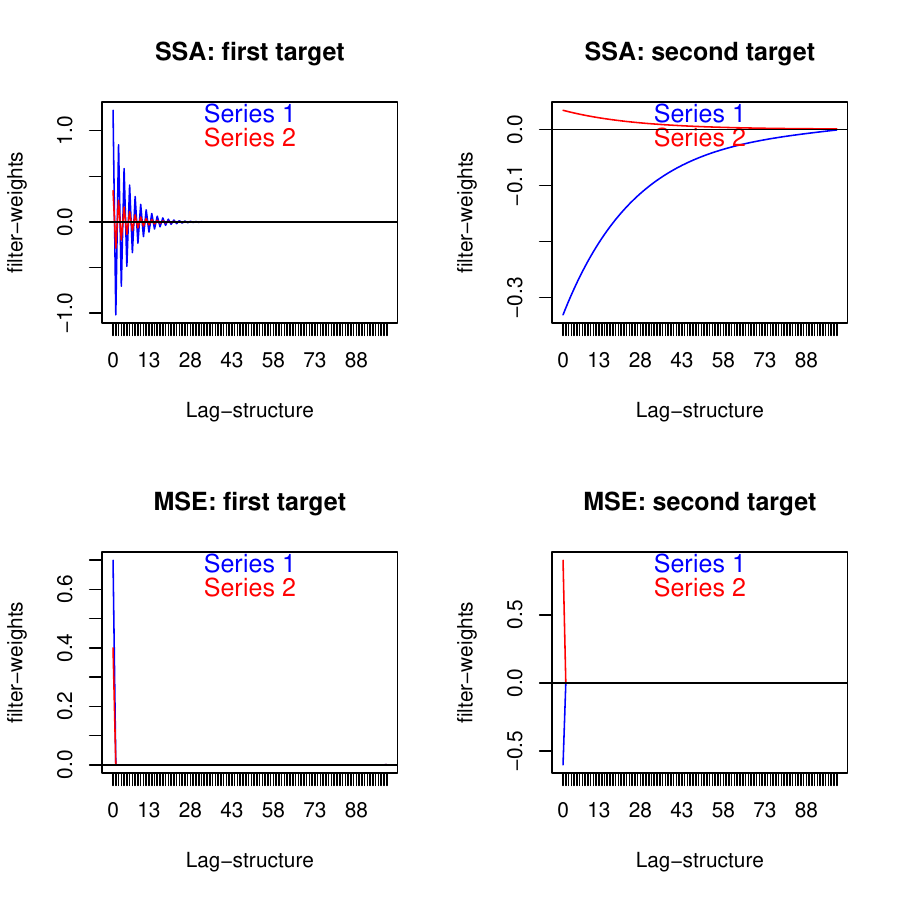}\caption{M-SSA  (top) and MSE (bottom) predictors for the first target series $z_{1t+\delta}$ (left) and the second target series $z_{2t+\delta}$ (right) with $\delta=1$ (one-step ahead). Predictor weights applied to first series $z_{1t}$ (blue) and second series $z_{2t}$ (red). The weights of the MSE predictor correspond to the first row of $A_1$ (bottom left) and the second row of  $A_1$ (bottom right), respectively.\label{filt_coef_var1}}\end{center}\end{figure}For confirmation, we proceed to compute sample estimates of the relevant performances measures for a  lengthy sample, $N=100000$, which ensures a near convergence of empirical to expected measures, as presented in Table \eqref{perf_var1}. The final column of the table corroborates the HTs of the MSE predictor, as derived previously. The optimal values of $\nu_i$ specified in Corollary \eqref{lambda_num_gen_mult} are $\nu_1$=-2.034$\in]-\infty,-2\rho_{max}(L)]$ (indicating unsmoothing) and $\nu_2$=2.001$ \in [2\rho_{max}(L),\infty[$ (indicating smoothing), as alluded to by the remark after the proof of Theorem \eqref{lambda_mult}.  
\begin{table}[ht]
\centering
\begin{tabular}{rrrrrr}
  \hline
 & Sample crit. & True crit. & Sample ht SSA & True ht SSA & Sample ht MSE \\ 
  \hline
Series 1 & 0.91 & 0.91 & 3.02 & 3.00 & 5.61 \\ 
  Series 2 & 0.66 & 0.67 & 8.04 & 8.00 & 4.65 \\ 
   \hline
\end{tabular}
\caption{Performances of SSA- and MSE-predictors: empirical and true holding-times as well as criterion values (correlations of SSA with MSE-predictors).  } 
\label{perf_var1}
\end{table}A comparative analysis of the sample predictors and targets is presented in Figure \eqref{output_var1}, which demonstrates the effects of `unsmoothing' in the first series and smoothing in the second series, as dictated by the HT constraints.

\begin{figure}[H]\begin{center}\includegraphics[height=3in, width=6in]{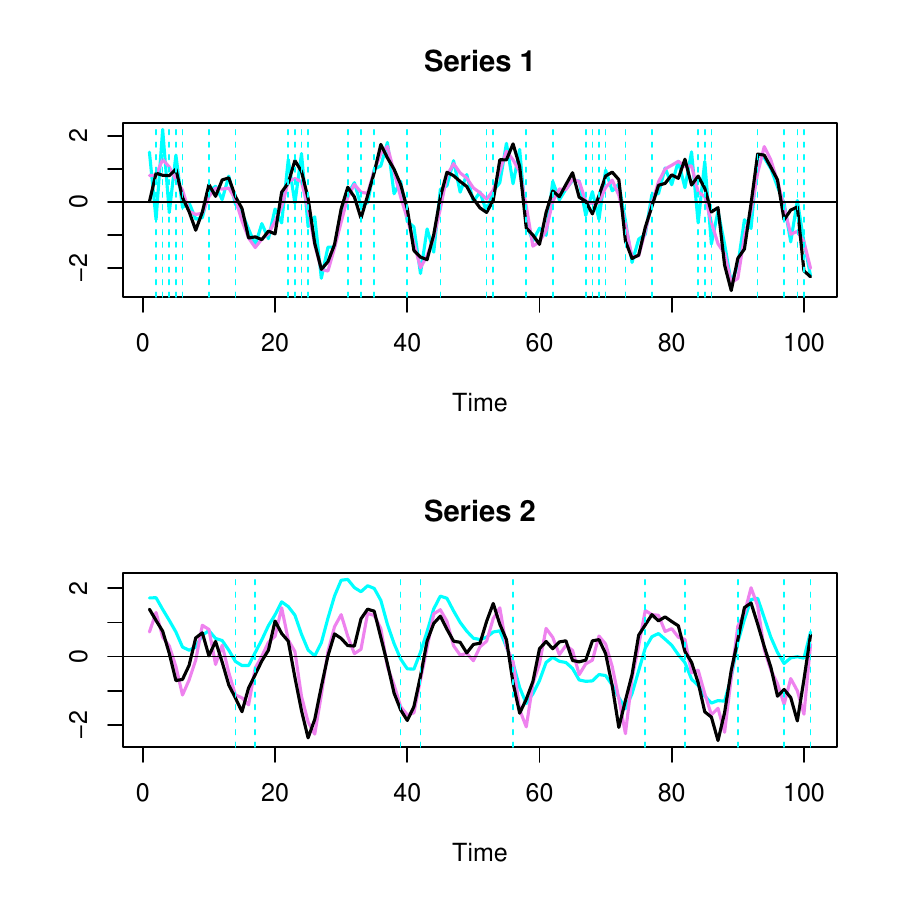}\caption{A comparison of filter outputs: M-SSA (cyan), MSE (violet) and target (black) for $\delta=1$. Zero-crossings by M-SSA  are marked by vertical lines.\label{output_var1}}\end{center}\end{figure}

\subsection{Smoothing}\label{smoo_app}

If the target correlation in the previous example is computed against the \emph{acausal} target $\mathbf{z}_{t+1}$, then M-SSA functions as a forecasting method. However, as discussed in Section \eqref{equi_s}, the MSE predictor can be utilized in place of the effective target without impacting the optimization results. In this context, M-SSA is then considered a smoother for the \emph{causal} MSE predictor, with smoothing referring specifically to the application of M-SSA to causal targets. In this context, we can examine the well-established Whittaker-Henderson (WH) graduation, expressed as follows:
\begin{eqnarray}\label{whithend}
\min_{\mathbf{b}_i}\sum_{t=L}^T(x_t-y_t)^2+\lambda\sum_{t=L+D}^T (\Delta^D y_t)^2,
\end{eqnarray}
where $\Delta^D$ represents the common difference operator  applied $D$-times. Notably, for $D=2$, this formulation yields the Hoddrick-Prescott (HP) filter, as detailed by Hodrick and Prescott (1997). The equation delineates a tradeoff between conflicting objectives of data-fitting and data-smoothing terms, with the regularization parameter $\lambda$ serving to balance this trade-off. 
An alternative formulation (dilemma) can be formalized within the (univariate) SSA framework as follows:
\begin{eqnarray}\label{crit_smooth}
\left.\begin{array}{cc}
&\max_{\mathbf{b}}\rho\Big(y,x,\delta\leq 0\Big)\\
&\rho(y)=\rho_1
\end{array}\right\},
\end{eqnarray}
where, in contrast to the previous section, $\delta\leq 0$. In this setting, the target $z_{t+\delta}=x_{t+\delta}$
corresponds to the output of a \emph{causal} back-shifting all-pass filter. 
If $L$ is defined as an odd integer and we select $\delta=-(L-1)/2$, this ensures symmetry in the backcast; the coefficients of the causal filter are centered at $x_{t-(L-1)/2}$, with the tails reflecting symmetrically around this central point. For the range $-(L-1)/2\leq \delta\leq 0$, the design transitions from symmetry to asymmetry as it approaches the nowcast at $\delta=0$. To illustrate, Fig.\eqref{wh_ssa_smoothing_all} presents the (univariate) SSA smoothers $\mathbf{b}_{\delta}$ based on Criterion \eqref{crit_smooth}, scaled arbitrarily to unit length and plotted as a function of $\delta$. 
This analysis assumes WN data, with $L=201$ and $\rho_1=0.9986$, which is inserted  into the HT constraint of Criterion \eqref{crit_smooth}. Here, $\rho_1$ represents the first-order ACF of the two-sided HP filter with parameter $\lambda=14400$, as derived from the WH criterion \eqref{whithend} and utilized in the subsequent Section \eqref{se_app}. This selection of $\lambda$ aligns with applications to monthly time series, as discussed by Ravn and Uhlig (2002). \\

\begin{figure}[H]\begin{center}\includegraphics[height=3in, width=4in]{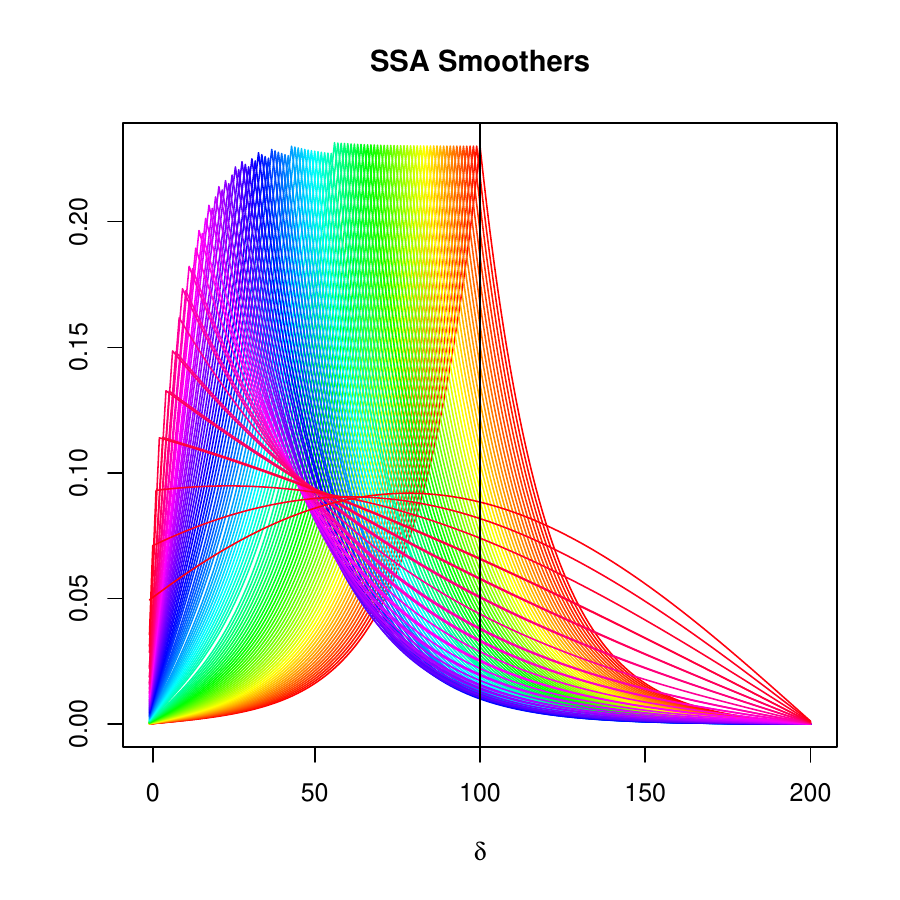}\caption{SSA smoothers scaled to unit length as a function of $\delta$ with values ranging from $\delta=-(L-1)/2=-100$ (symmetric filter) to $\delta=0$ (asymmetric nowcast), assuming the data to be WN and imposing the HT (or first-order ACF) of the two-sided HP(14400) filter. \label{wh_ssa_smoothing_all}}\end{center}\end{figure}The SSA smoothers in this example replicate exactly the HT of HP, regardless of $\delta$. Consequently, by Theorem \eqref{lambda_mult}, they are expected to produce superior target correlations or achieve lower MSEs compared to HP (assuming white noise data). Conversely, for similar target correlation, we expect SSA to yield superior HT when compared to HP, owing to Theorem \eqref{cor3} (duality). For illustration, Figure \eqref{wh_ssa_smoothing} presents a truncated symmetric causal HP(14400) filter alongside two distinct symmetric SSA smoothers, with the assumption that $\delta = -\frac{(L-1)}{2}$. Additionally, Table \eqref{smooth_comp} provides a comparative analysis of the performance of all three designs, assuming the data follows a WN sequence.
\begin{table}[ht]
\centering
\begin{tabular}{rrrr}
  \hline
 & HP & SSA1 & SSA2 \\ 
  \hline
Holding times & 59.548 & 59.548 & 75.000 \\ 
  Target correlations & 0.205 & 0.228 & 0.205 \\ 
  RMS second-order differences & 0.005 & 0.024 & 0.017 \\ 
   \hline
\end{tabular}
\caption{HP vs. two different SSA smoothers. SSA1 replicates the holding time of HP and SSA2 replicates its target correlation. Root mean-squared second-order differences in the last row refer to standardized white noise data.} 
\label{smooth_comp}
\end{table}For identical HT, the SSA1 design (represented by the blue line in the figure) demonstrates superior performance in target correlation compared to HP, as claimed. Similarly, the SSA2 design (depicted by the violet line in the figure) exhibits a maximal HT, surpassing the HP filter, given an equivalent target correlation. However, when evaluating smoothness based on the curvature (mean-squared second order differences), HP outperforms both SSA smoothers. 
The observed variations across the reported performance metrics are sufficiently substantial to warrant a carefully informed decision between minimizing the rate of zero-crossings with SSA or reducing curvature via WH. In summary, the simplified SSA optimization problem \eqref{crit_smooth} can be regarded as a novel smoothing algorithm that sets itself apart from traditional methods such as the WH or HP approaches derived from Criterion \eqref{whithend}. This distinction stems from the fact that the regularization term is explicitly formulated as a HT constraint, providing an alternative perspective on controlling smoothness in time series analysis. We argue that regulating the rate of sign changes offers a viable alternative to classical filtering and smoothing techniques, particularly within economic contexts characterized by alternating growth phases, where it aligns more closely with the underlying decision-making and control logic, as discussed in Wildi (2024).\\

\begin{figure}[H]\begin{center}\includegraphics[height=3in, width=4in]{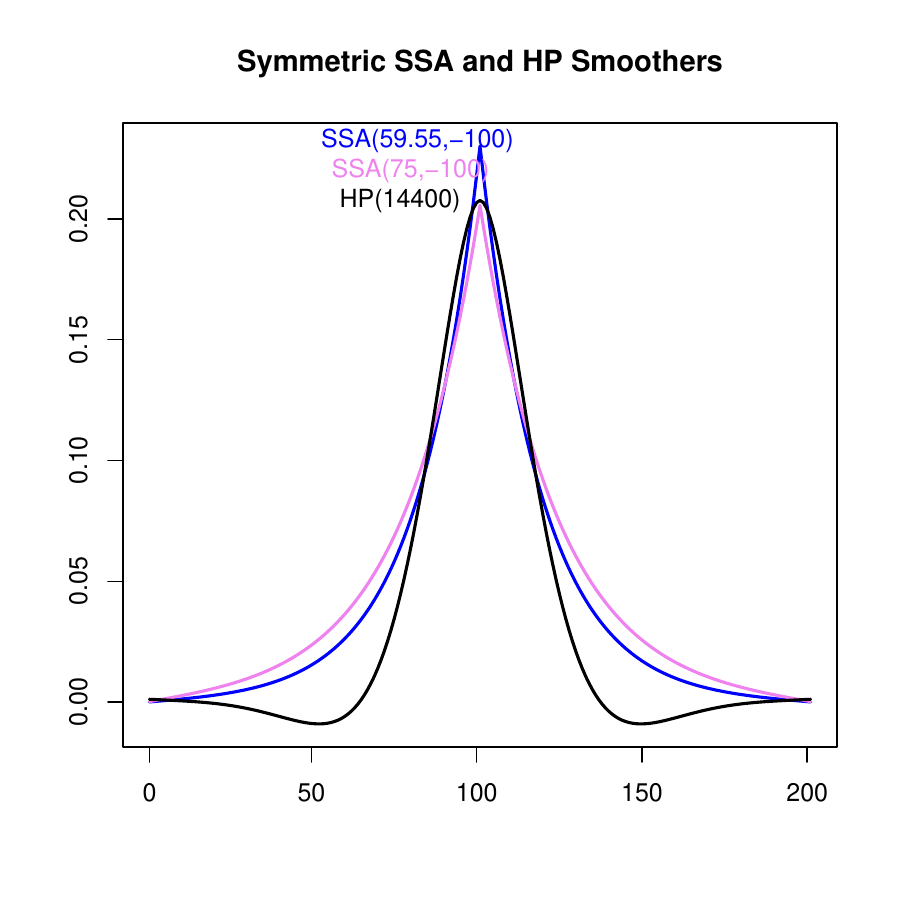}\caption{Coefficients of symmetric causal SSA and HP smoothers of length 201, arbitrarily scaled to unit length: the first SSA design (blue line) replicates the holding time of HP, the second SSA design (violet line) replicates the tracking-ability or target correlation of HP.\label{wh_ssa_smoothing}}\end{center}\end{figure}

The univariate smoothers in the above example are determined by the assumption of WN. We now extend the framework to multivariate autocorrelated data, as discussed in Section \eqref{ext_stat}.  
For illustration, 
consider the following three-dimensional VAR(1),
\[
\mathbf{x}_t=\left(\begin{array}{ccc}
0.7&0.4&-0.2\\
-0.6&0.9&0.3\\
0.5&0.2&-0.3
\end{array}\right)\mathbf{x}_{t-1}+\boldsymbol{\epsilon}_t~,~\boldsymbol{\Sigma}=\left(\begin{array}{ccc}
3.17&0.77&-0.5\\
0.77&0.69&0\\
-0.5&0&1.7
\end{array}\right)
\]
with Wold-decomposition $\boldsymbol{\Xi}_k=\mathbf{A}^k$, assuming a causal allpass target $z_t=x_t$ with $\delta=0$, which corresponds to the proper identity. 
We select $L=51$ and set the HTs $ht_1=8$, $ht_2=6$ and $ht_3=10$ within the HT constraints for the three M-SSA smoothers $y_{1t},y_{2t}$ and $y_{3t}$, as illustrated in Fig.\eqref{filt_coef_alternative_signal_extraction}. 
The plots indicate that the second series $x_{2t}$ (in red) serves as a significant explanatory variable for all three target series (after suitable adjustment of scales). Fig.\eqref{peak_corr} provides insight into this phenomenon, as the peaks of the CCF shown in the left panel, along with the original time series in the right panel, indicate that $x_{2t}$ is left-shifted or leading. This finding can be interpreted as supporting the importance of a leading indicator within a multivariate nowcasting framework. Table \eqref{tab_1} presents expected and sample performance measures, the latter in parentheses, based on an extensive sample size $N=100000$, which corroborates the posited constraints. Additionally, the HTs derived from the original (unfiltered) data, detailed in the last column, identify the second series as  the smoothest  among the three. 
The double strike of  a lead as well as an increased smoothness of $x_{2t}$ further reinforce its importance as an explanatory variable in the context of endpoint smoothing (or forecasting). \\

\begin{figure}[H]\begin{center}\includegraphics[height=3in, width=6in]{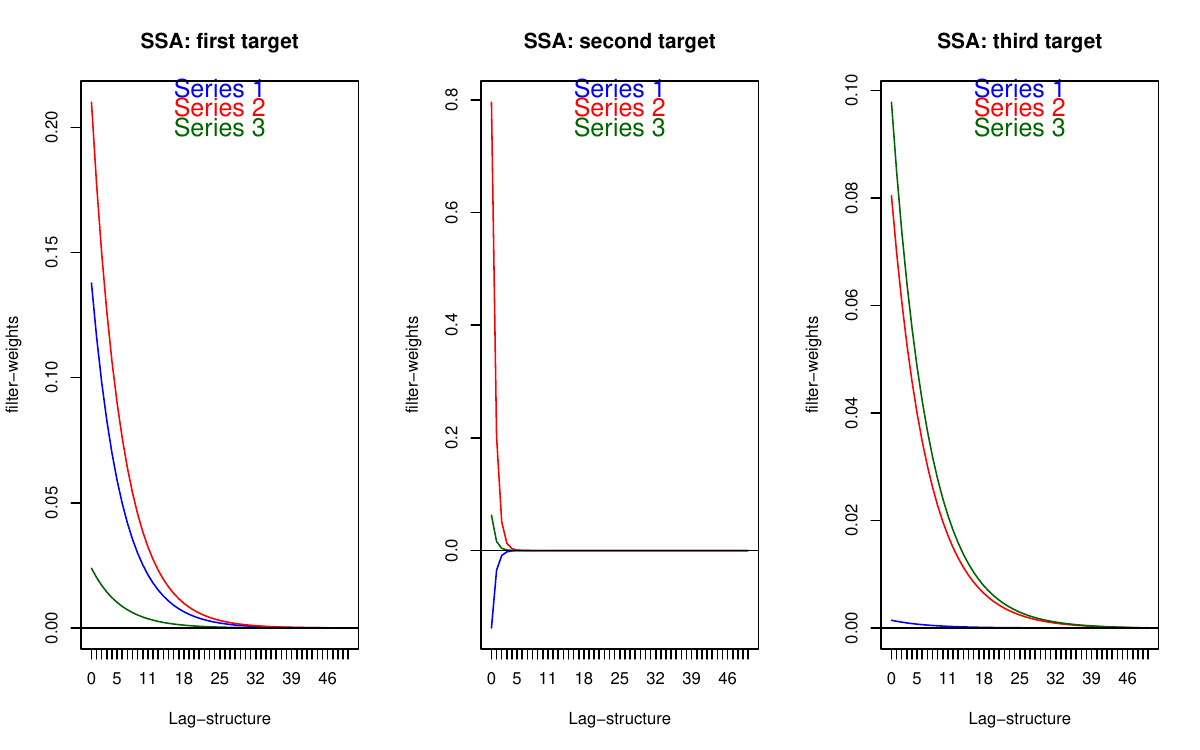}\caption{M-SSA smoother for $\delta=0$: first target series (left), second target  (middle) and third target (right). Filters applied to first series (blue), second series (red) and third series (green). \label{filt_coef_alternative_signal_extraction}}\end{center}\end{figure}\begin{figure}[H]\begin{center}\includegraphics[height=3in, width=6in]{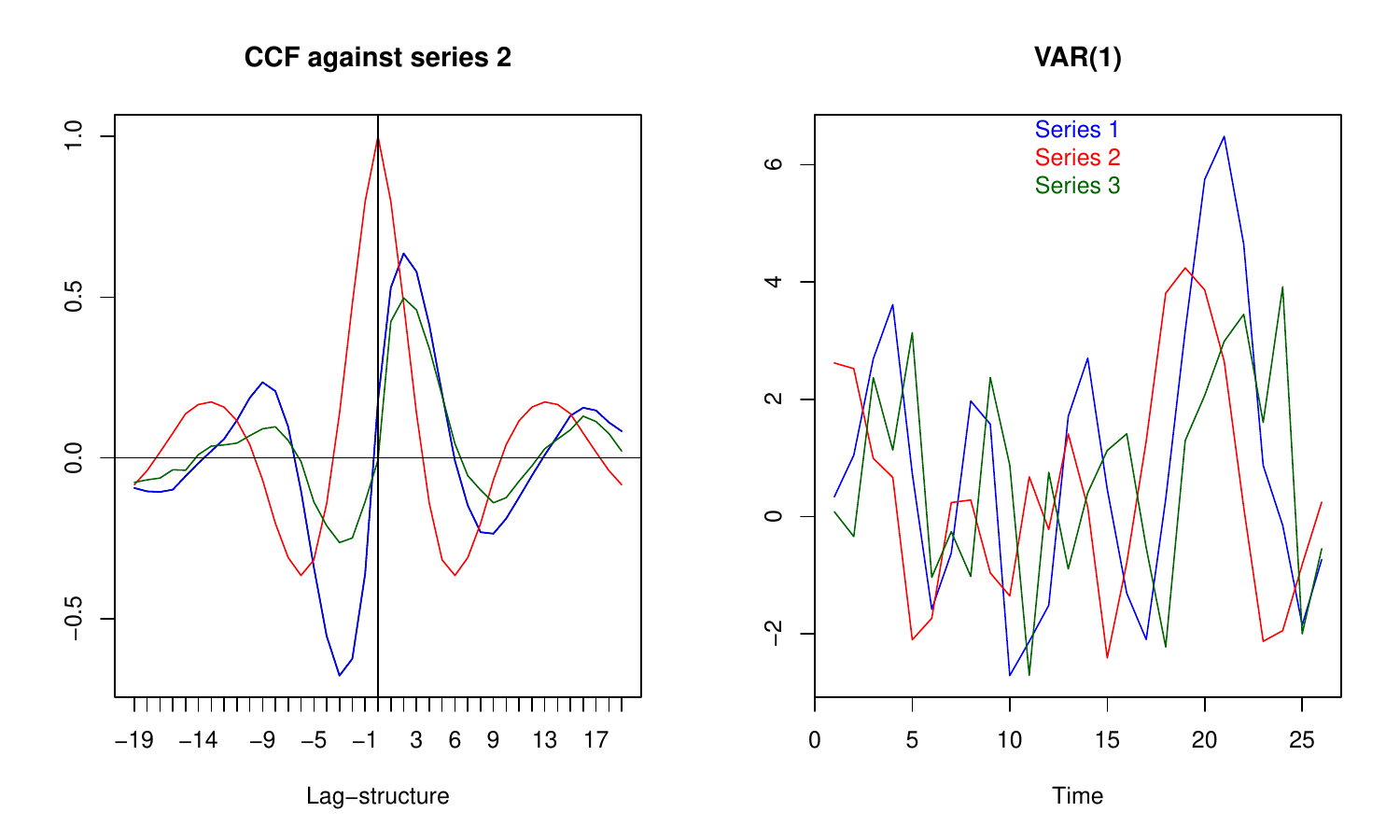}\caption{Left panel: CCF of second series and first series (blue), ACF of second series (red) and CCF of second series and  third series (green): the right-shifted peaks of the CCFs (blue and green) suggest that the second series $x_{2t}$ is a leading time series. Right panel: a short subsample of $x_{it}$, $i=1,2,3$: the second series (red) appears left-shifted (leading).\label{peak_corr}}\end{center}\end{figure}
\begin{table}[ht]
\centering
\begin{tabular}{rllll}
  \hline
 & Sign accuracy & Cor. with data & HT M-SSA & HT of data \\ 
  \hline
1 & 0.74 (0.74) & 0.69 (0.69) & 8 (8.01) & 3.91 (3.89) \\ 
  2 & 0.96 (0.95) & 0.99 (0.99) & 6 (6) & 4.9 (4.88) \\ 
  3 & 0.66 (0.66) & 0.48 (0.48) & 10 (10) & 2.12 (2.11) \\ 
   \hline
\end{tabular}
\caption{Expected performances of M-SSA (first three columns) and expected HTs of original data $\mathbf{x}_t$ (last column). Sample estimates in parentheses are based on an extensive sample of size $N=100000$. } 
\label{tab_1}
\end{table}Fig.\eqref{output_alternative_signal_extraction} presents a comparison between  $x_{it}$ (black) and M-SSA endpoint smoothers $y_{it}$ (cyan), for $i=1,2,3$, thereby confirming the previously discussed findings.
Specifically, the smoothing task proves to be significantly less challenging for $x_{2t}$ relative to $x_{3t}$, as evidenced by the more substantial difference between the HTs of the original data and the smoothers in the latter case. This observation is further supported by the last two columns of Table \eqref{tab_1}. It is worth noting that if we selected $\delta=-(L-1)/2$ in the above example, instead of $\delta=0$, the resulting smoothers $y_{it}$ would simplify to (nearly) symmetric and univariate designs (not shown).

\begin{figure}[H]\begin{center}\includegraphics[height=3in, width=6in]{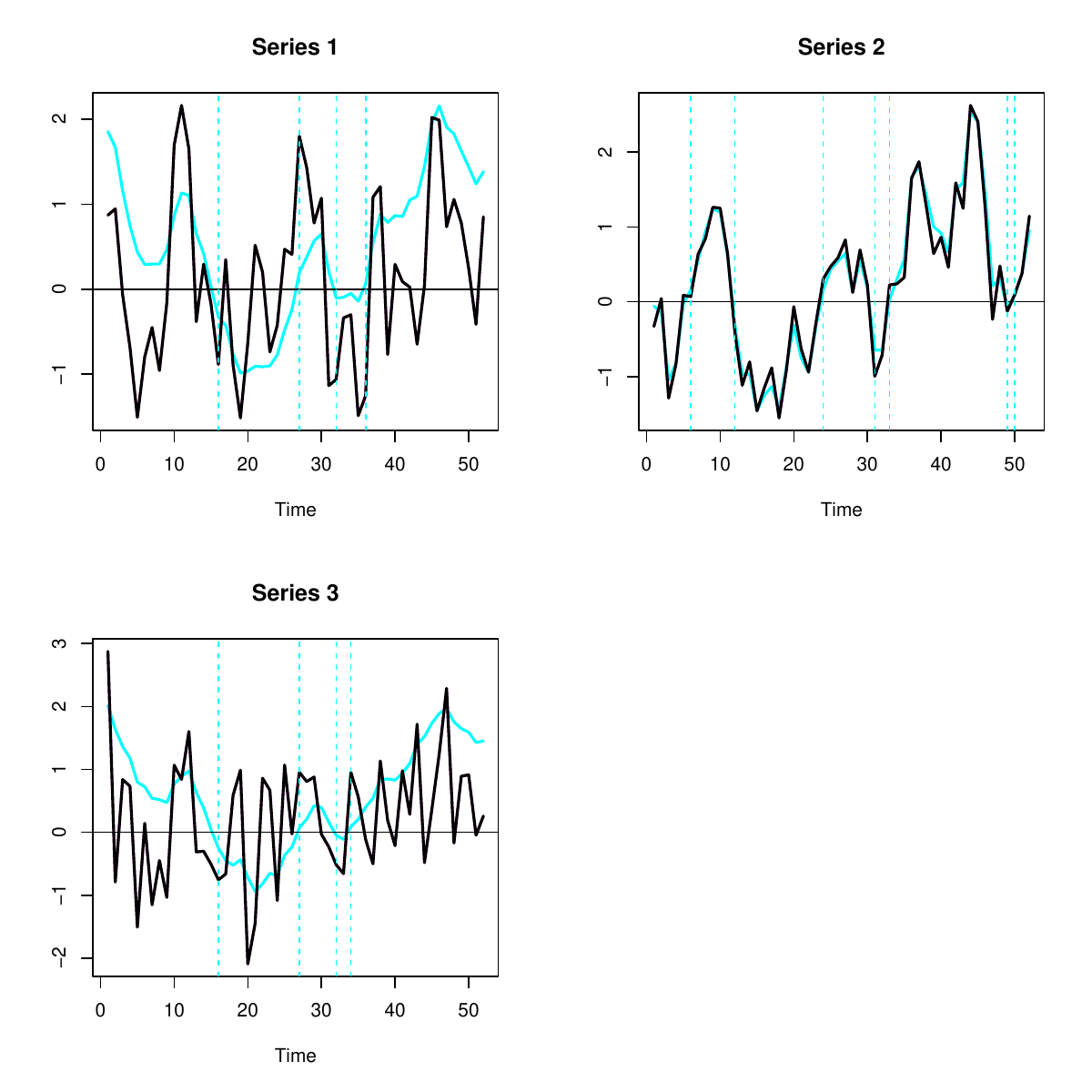}\caption{A comparison of original data (black) and M-SSA smoothers (cyan). Zero-crossings of the smoothers are marked by vertical lines.\label{output_alternative_signal_extraction}}\end{center}\end{figure}

\subsection{Signal Extraction}\label{se_app}

The M-SSA criterion is capable of addressing both uni- and multivariate forecasting and smoothing tasks.  Here, we consider signal extraction, allowing for both the target and the MA inversion to deviate from allpass (identity) designs, representing the highest level of complexity for M-SSA. To illustrate this, we follow the business-cycle application outlined in Wildi (2024), which applies the HP(14400) filter to the monthly U.S. industrial production index\footnote{Board of Governors of the Federal Reserve System (US), Industrial Production: Total Index [INDPRO], retrieved from FRED, Federal Reserve Bank of St. Louis; https://fred.stlouisfed.org/series/INDPRO, October 23, 2024.}. We extend the scope of this analysis by incorporating an additional leading indicator, the Composite Leading Indicator (CLI) provided by the OECD\footnote{Organization for Economic Co-operation and Development, Leading Indicators OECD: Leading Indicators: Composite Leading Indicator: Amplitude Adjusted for United States [USALOLITOAASTSAM], retrieved from FRED, Federal Reserve Bank of St. Louis; https://fred.stlouisfed.org/series/USALOLITOAASTSAM, October 23, 2024.}, as depicted in the top left panel of Fig.\eqref{data}.
\begin{figure}[H]\begin{center}\includegraphics[height=3in, width=6in]{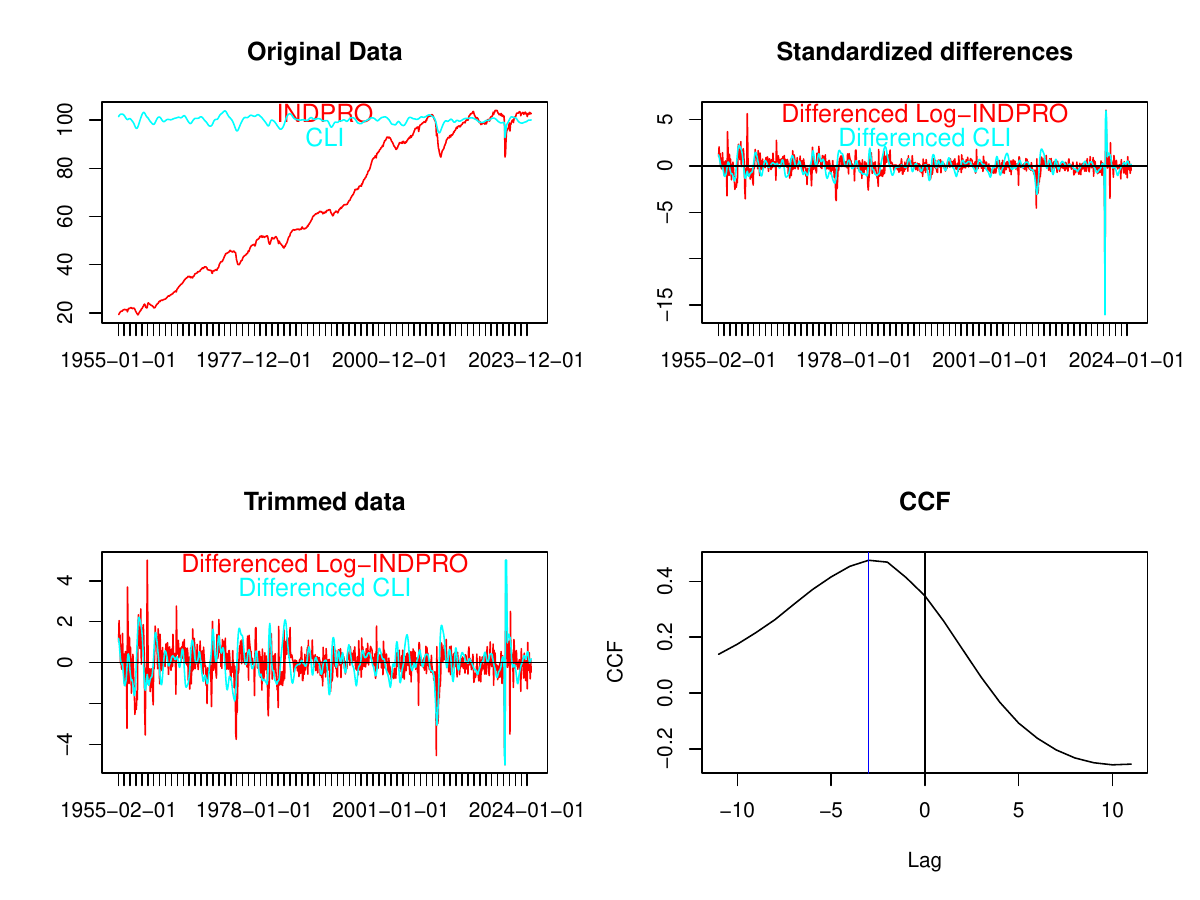}\caption{The monthly Industrial Production Index (INDPRO) and Composite Leading Indicator (CLI) are presented in the top left panel, with standardized log-differences displayed in the top right panel. In the bottom left panel, singular pandemic outliers have been removed to ensure that absolute values do not exceed $5\sigma$. The CCF of log-differences is shown in the bottom right panel. The  peak of the CCF at lag -3, suggests an advancement or lead of the CLI by one quarter.\label{data}}\end{center}\end{figure}Given that the CLI fluctuates around a fixed level, we now consider the first differences of the log-transformed data.  The top right panel of the figure presents the corresponding standardized time series, while the bottom left panel displays the actual data utilized for fitting. In this latter panel, singular pandemic outliers have been removed (trimmed) to ensure that absolute values do not exceed $5\sigma$, a threshold that aligns with extreme observations from earlier historical data. The peak of the CCF in the lower right panel indicates that the CLI leads the INDPRO by one quarter. 
We examine the application of the HP trend filter to differenced series, with a focus on the growth cycle to track U.S. recession episodes as identified by the National Bureau of Economic Research (NBER)\footnote{In contrast, the conventional output gap is defined as the difference between the actual (original non-stationary) time series and its trend, which addresses deviations from `potential output'. A comparative analysis of both designs is presented in Wildi (2024). However, to conserve space, this discussion will utilize the single growth cycle concept to illustrate the M-SSA.}. We then compare univariate and bivariate designs, highlighting target correlations, holding times (HTs) and leads (left shifts/advancement) or lags (right-shifts/retardation) of the corresponding nowcasts. Specifically, besides the bivariate MSE nowcast we consider the univariate SSA and the classic concurrent HP, as discussed in McElroy (2006), denoted as HP-C, as additional benchmarks. \\

To ensure clarity and parsimony,  we employ a simple bivariate model, excluding data revisions by relying on the final releases of both indicators. Additionally, the model incorporates all available observations.\footnote{Sample truncation may affect model identification and parameter estimation, although the impact on the target correlation and the HT of the resulting filters remains limited unless truncation is severe (e.g., fewer than half of the available observations).} While a more comprehensive analysis would likely yield improved forecast performance and potentially insightful out-of-sample conclusions, our primary objective is to illustrate the accuracy-smoothness trade-off based on the M-SSA baseline extension discussed in Section \eqref{op_mod}, avoiding unnecessary complications from modeling issues pertinent to the baseline (pseudo-)MSE specification. Heinisch et al. (2026) present a more complex five-dimensional M-SSA predictor for German GDP, emphasizing explicit modeling considerations; the corresponding empirical framework is available in our open source M-SSA package for further investigation.

\subsubsection{Univariate: HP-C and  SSA}

We begin by deriving a univariate SSA framework, with the filter length set to  $L=201$, see Wildi (2024). 
This analysis yields the following ARMA(2,1) model for the transformed INDPRO series illustrated in the bottom left panel of Fig.\eqref{data}:
\begin{eqnarray}\label{arma_mod}
\mathbf{x}_t&=&0.01+0.96x_{t-1}-0.16x_{t-2}+\epsilon_t-0.64\epsilon_{t-1}.
\end{eqnarray}
According to standard diagnostic tests, the residiuals are white noise. This model is employed to derive the MA inversion $\xi_k$ for $k=0,...,L-1=200$, as illustrated in the left panel of Fig. \eqref{filt_coef_INDPRO_uni_x}. 
Our acausal objective $z_{t+\delta}$ utilizes the (truncated) two-sided Hodrick-Prescott filter with  $\lambda = 14400$, as illustrated in the middle panel of the figure. The (univariate) MSE nowcast of $z_{t+\delta}$, where  $\delta=0$, is derived following  McElroy and Wildi (2016)\footnote{\label{footmc}An open source R-package is provided by  McElroy and Livsey (2022) with further discussion available in  McElroy, T. (2022).}
\begin{eqnarray}\label{sem_inf_mse}
\hat{\gamma}_{x\delta}(B)=\sum_{k\geq 0}\gamma_{k+\delta}B^k+\sum_{k<0}\gamma_{k+\delta}\left[\boldsymbol{\xi}(B)\right]_{|k|}^{\infty}B^k\boldsymbol{\xi}^{-1}(B),
\end{eqnarray}
where $B$ denotes the backshift operator, $\boldsymbol{\xi}(B)=\sum_{k\geq 0}\xi_k B^k$, $\boldsymbol{\xi}^{-1}(B)$ represents the AR-inversion of the data generating process and $\gamma_{k}$, $|k|<\infty$, are the weights of the two-sided HP filter. The notation $[\cdot]_{|k|}^{\infty}$ signifies the omission of the first $|k|-1$ lags\footnote{The derivation of the MSE benchmark follows three distinct steps: first, compute the convolution of the (finite length) two-sided target filter $\gamma_k$
and the MA inversion $\xi_k$ 
of the ARMA model; second, truncate this expression by omitting the acausal terms associated with all future innovations $\epsilon_{t-k}$ for $k < 0$; lastly, deconvolute the MSE filter from the truncated convolution, as specified in Equation \eqref{con_inv}. The M-SSA package provides comprehensive guidance on the sequence of transformations and operations.}. 
The resulting filter is depicted in the right panel of Fig. \eqref{filt_coef_INDPRO_uni_x} (green line). \\

The MSE nowcast filter exhibits a first-order ACF of 0.963, which corresponds to a HT of 11.508. Similarly, the classic HP-C filter has first-order ACF 0.967, which corresponds to a HT of 12.267. The (univariate) SSA, as illustrated in Fig.\eqref{filt_coef_INDPRO_uni_x} (blue line, right panel), is derived from the MSE nowcast,  as outlined in  Wildi (2024) (the proceeding corresponds to Section \eqref{ext_stat} for the  univariate case). Within the HT constraint, we set $ht_1=17.26$, which surpasses the MSE benchmark (and incidentally HP-C) by 50$\%$.  As a result, we anticipate that the SSA nowcast will yield a reduction of 33$\%$ in zero-crossings over the long term. While this choice of the HT constraint is somewhat arbitrary, our selection is mainly intended to illustrate the accuracy–smoothness trade-off. 
\begin{figure}[H]\begin{center}\includegraphics[height=3in, width=6in]{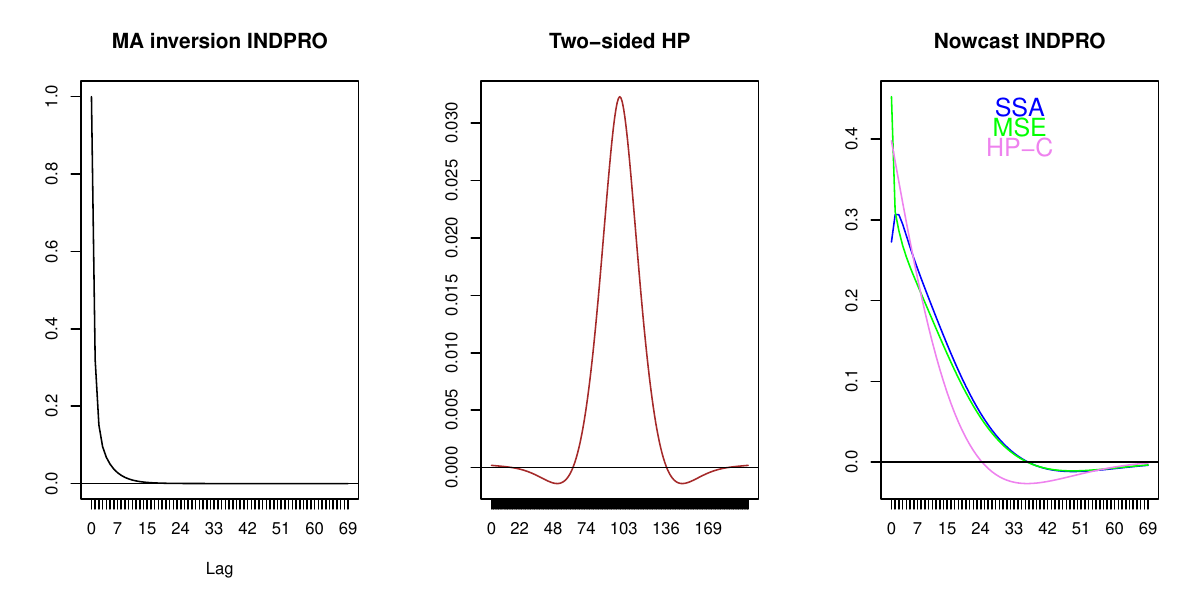}\caption{MA inversion of ARMA-model fitted to INDPRO (left panel). Truncated two-sided HP target filter (middle panel). Univariate MSE (green), HP-C (violet) and SSA nowcast filters (blue) (right panel). Only the first 70 lags are shown for clarity, and all filters are scaled to unit length to facilitate visual comparison.\label{filt_coef_INDPRO_uni_x}}\end{center}\end{figure}A comparison of SSA and HP-C nowcasts in Fig.\eqref{INDPRO_uni_HP_out} reveals that the SSA  exhibits increased smoothness, albeit at the cost of a one-month rightward shift or retardation relative to HP-C (sample performances are reported in Tables  \eqref{tab_bi_1} and  \eqref{tab_bi_2}). Thus , we anticipate that the bivariate leading indicator design, which incorporates the additional CLI, will mitigate this lag while maintaining comparable smoothness.
\begin{figure}[H]\begin{center}\includegraphics[height=3in, width=6in]{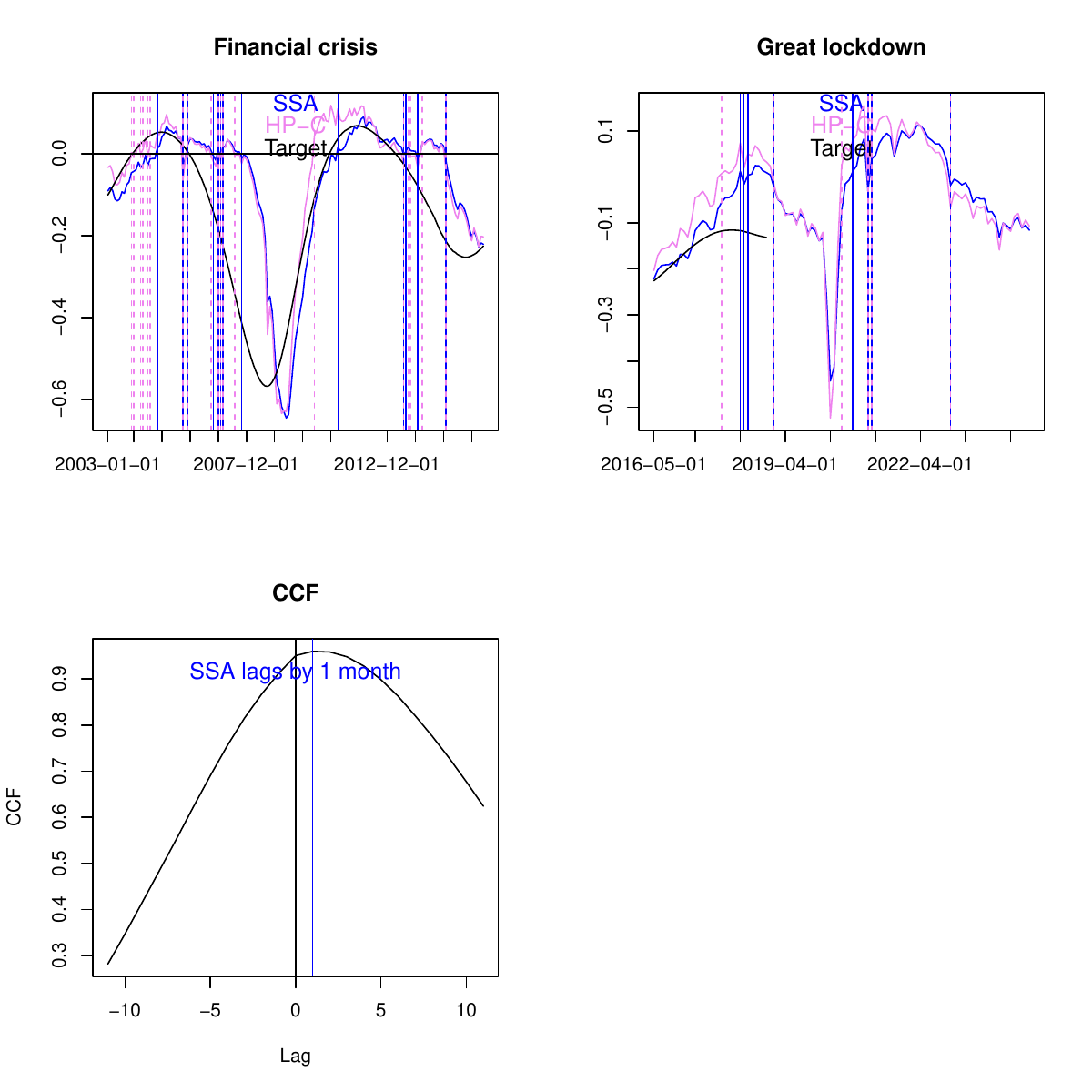}\caption{SSA (blue) vs. HP-C (violet), illustrating their performance during the financial crisis (top left) and the pandemic (top right). Zero-crossings are indicated by vertical lines with corresponding colour-coding: the SSA exhibits less crossings (smoother). The cross-correlation function presented in the bottom panel suggests that SSA exhibits a one-month lag relative to HP-C.\label{INDPRO_uni_HP_out}}\end{center}\end{figure}

\subsubsection{M-SSA: Bivariate Leading Indicator Design}

We examine a bivariate design that incorporates the CLI as an additional explanatory variable for the industrial production index. The following VARMA(3,1) model is specified as:
\begin{eqnarray*}
\mathbf{x}_t&=&\left(\begin{array}{c}0\\0\end{array}\right)+\left
(\begin{array}{cc}0.63&0.32\\
-0.28&1.28\end{array}\right)\mathbf{x}_{t-1}+
\left(\begin{array}{cc}-0.07&-0.44\\
-0.05&-0.36\end{array}\right)\mathbf{x}_{t-2}+
\left(\begin{array}{cc}0.02&0.3\\
0&0.09\end{array}\right)\mathbf{x}_{t-3}\\
&&+
\boldsymbol{\epsilon}_t+\left(\begin{array}{cc}0.5&-0.43\\
-0.19&0.2\end{array}\right)\boldsymbol{\epsilon}_{t-1}
\end{eqnarray*}
with the covariance matrix given by $\boldsymbol{\Sigma} = \left(\begin{array}{cc}0.562&0.05414\\ 0.05414&0.1494\end{array}\right)$, where the first series corresponds to INDPRO. 
The bivariate residuals are confirmed to exhibit white noise characteristics according to standard diagnostic tests. Subsequently, we focus attention on a INDPRO nowcast, discarding a corresponding analysis for the CLI. For consistency, we apply the same HT constraint as for the SSA in the previous section. 
\begin{table}[ht]
\centering
\begin{tabular}{rrrrrrr}
  \hline
 & Cor. HP-C & Cor. M-SSA & Cor. MSE & HT HP-C & HT M-SSA & HT MSE \\ 
  \hline
Expected & 0.650 & 0.736 & 0.744 & 11.132 & 17.263 & 11.011 \\ 
  Sample & 0.649 & 0.734 & 0.743 & 11.120 & 17.180 & 10.947 \\ 
   \hline
\end{tabular}
\caption{Simulation experiment: convergence of sample performances to expected numbers based on a single long sample of length one Million of the bivariate VARMA(3,2)-model. The first three columns correspond to target correlations (correlations of nowcasts and acausal HP trend) and the last three columns are the holding times of the nowcasts.} 
\label{tab_bi_0}
\end{table}A simulation experiment utilizing an extensive sample of one million observations from the aforementioned VARMA process corroborates the convergence of sample performance metrics to their expected values, as detailed in Table \eqref{tab_bi_0}. The M-SSA exhibits the largest HT, which is $50\%$ greater than that of the MSE predictor (as constrained), and the second-largest target correlation, being only marginally outperformed by the MSE predictor in this dimension. It also outperforms the univariate benchmark HP-C on both accounts. 
It is important to emphasize that the results presented in the table are contingent upon the bivariate VARMA model. Specifically, for the classic HP-C the values would differ if the univariate ARMA model from the previous section were utilized instead.  The derivation of these `true' performance metrics, presented in the first row of the table, is discussed in the appedix. In any case, the M-SSA package provides comprehensive guidance on the sequence of operations and transformations necessary to replicate all results. 
Note also that we exclude the univariate MSE from the subsequent benchmark comparisons since it is outperformed by the bi-variate (MSE) design.\\

The MA-inversion (impulse response function) of the VARMA-model, along with the nowcast weights, is presented in Fig.\eqref{MA_inv_VARMA}. They indicate that the CLI serves as an important explanatory variable for INDPRO. The rightward shift of the majority of the CLI MA-weights (top left) further substantiates the leading characteristic of the indicator. 
\begin{figure}[H]\begin{center}\includegraphics[height=3in, width=6in]{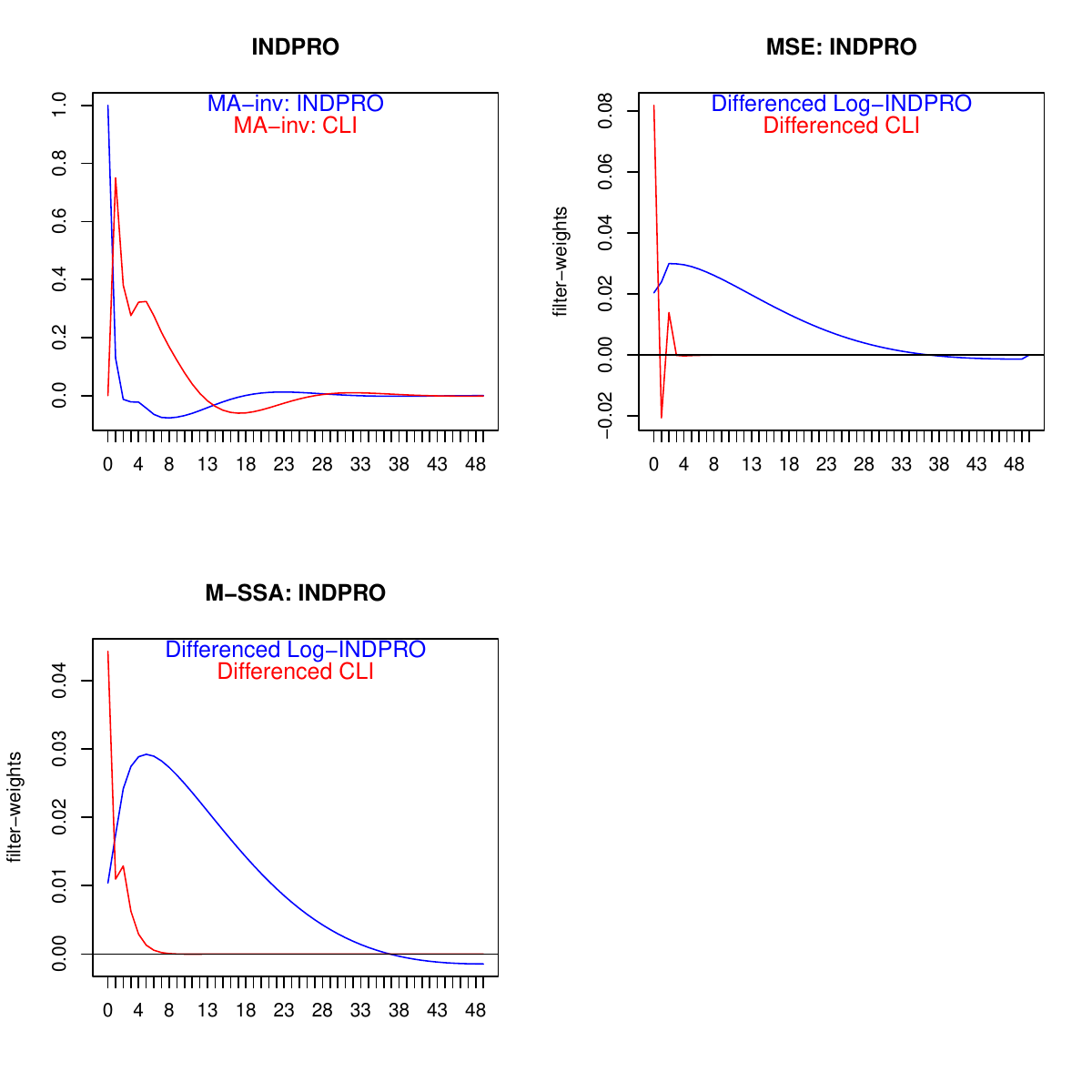}\caption{MA-inversion of VARMA(3,1) model (top left). MSE (top right) and M-SSA nowcasts (bottom) for INDPRO: all series are truncated to length 50 to facilitate visual comparison.\label{MA_inv_VARMA}}\end{center}\end{figure}The sample and expected performances of the classic univariate HP-C filter and the univariate SSA, alongside the bivariate MSE and M-SSA nowcasts are compared in Tables \eqref{tab_bi_1} 
(target correlations) and \eqref{tab_bi_2} (HTs), with standard errors indicated in parentheses. All expected values are based on the bivariate VARMA model, except for the SSA (in the second columns of the tables), which are derived from the univariate ARMA model (therefore, direct comparisons of expected numbers with the SSA should be approached with caution). The sample estimates are derived from data spanning  from 1955-02-01 to  2024-09-01, resulting in a net number of observations of $N=605$\footnote{We utilize the full available common sample of INDPRO and CLI to mitigate the impact of excessive variance in the sample HT. The net sample length is calculated as $T-L-L^{left-tail}=836-201-30=605$, where $L$ represents the filter length (initialization for all designs) and $L^{left-tail}$ denotes the length of the left `acausal' side of the target, thereby restricting comparisons to time points $t \leq T-L^{left-tail}$ near the end of the sample.}. A comparison of expected and sample HTs  indicates that the latter are consistently larger, with a statistically significant finding for the MSE nowcast, partially due to the occurrence of protracted recession episodes\footnote{Alternative models employing shorter subsamples partially alleviate this issue without impacting our primary conclusions; thus, they are not included in this discussion.}. 
The SSA and M-SSA designs exhibit identical expected HTs by design (noting, however, that the underlying models differ). Nonetheless, sample HTs highlight the M-SSA as the smoothest nowcast, characterized by the fewest zero-crossings. This finding is further supported by the corresponding filter outputs displayed in Fig. \eqref{INDPRO_CLI_multi_out}. 
The M-SSA demonstrates superior sample performances compared to the univariate designs, both in terms of target correlation and HT. Additionally, the CCF depicted in Fig. \eqref{comparison_bivariatecli_univariate} (bottom right panel) shows that M-SSA and HP-C are now coincident, unlike the SSA discussed in the previous section, which displayed a lagging behavior. The figure further highlights a two-month lead of the M-SSA design over the univariate SSA (top right). Finally, the bivariate MSE nowcast outperforms in terms of expected and sample target correlations, noting that the expected value of the SSA in Table \eqref{tab_bi_1} cannot be directly compared, due to differing underlying models. \\

\begin{table}[ht]
\centering
\begin{tabular}{rllll}
  \hline
 & Cor. HP-C & Cor. SSA & Cor. M-SSA & Cor. MSE \\ 
  \hline
Expected & 0.65 & 0.755 & 0.736 & 0.744 \\ 
  Sample & 0.735(0.074) & 0.717(0.085) & 0.77(0.077) & 0.791(0.073) \\ 
   \hline
\end{tabular}
\caption{Expected and sample target correlations of univariate HP-C and SSA (first two columns) and bivariate M-SSA and MSE nowcasts, with standard deviations in parentheses.} 
\label{tab_bi_1}
\end{table}
\begin{table}[ht]
\centering
\begin{tabular}{rllll}
  \hline
 & HT HP-C & HT SSA & HT M-SSA & HT MSE \\ 
  \hline
Expected & 11.132 & 17.263 & 17.263 & 11.011 \\ 
  Sample & 15.512(2.183) & 18.508(6.629) & 24.462(4.922) & 19.875(3.081) \\ 
   \hline
\end{tabular}
\caption{Expected and sample HTs of univariate HP-C and SSA (first two columns) and bivariate M-SSA and MSE nowcasts, with standard deviations in parentheses.} 
\label{tab_bi_2}
\end{table}\begin{figure}[H]\begin{center}\includegraphics[height=3in, width=6in]{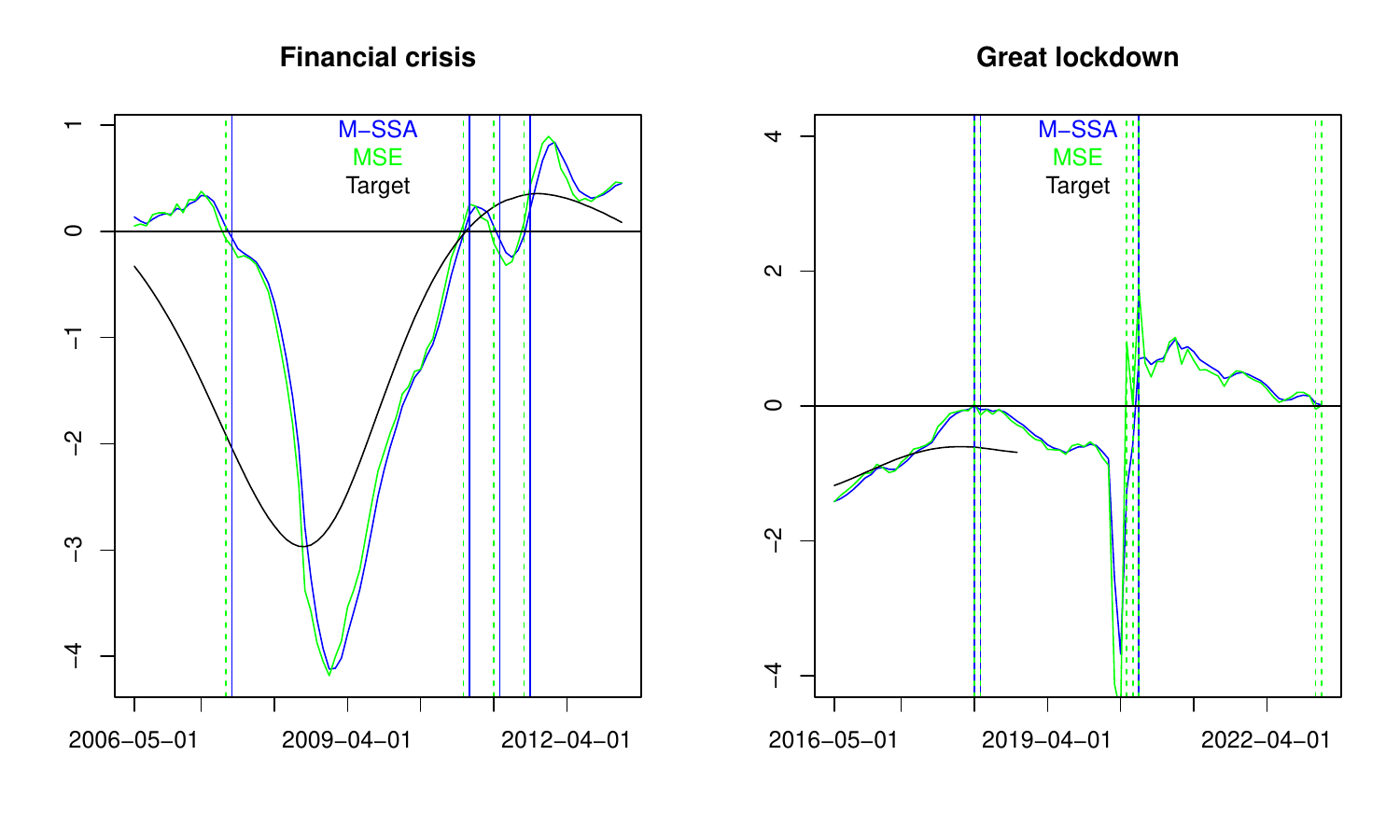}\caption{Bivariate M-SSA (blue), MSE (green) and acausal target (black) with zero-crossings of M-SSA (vertical blue lines) and MSE (vertical green lines).\label{INDPRO_CLI_multi_out}}\end{center}\end{figure}\begin{figure}[H]\begin{center}\includegraphics[height=3in, width=6in]{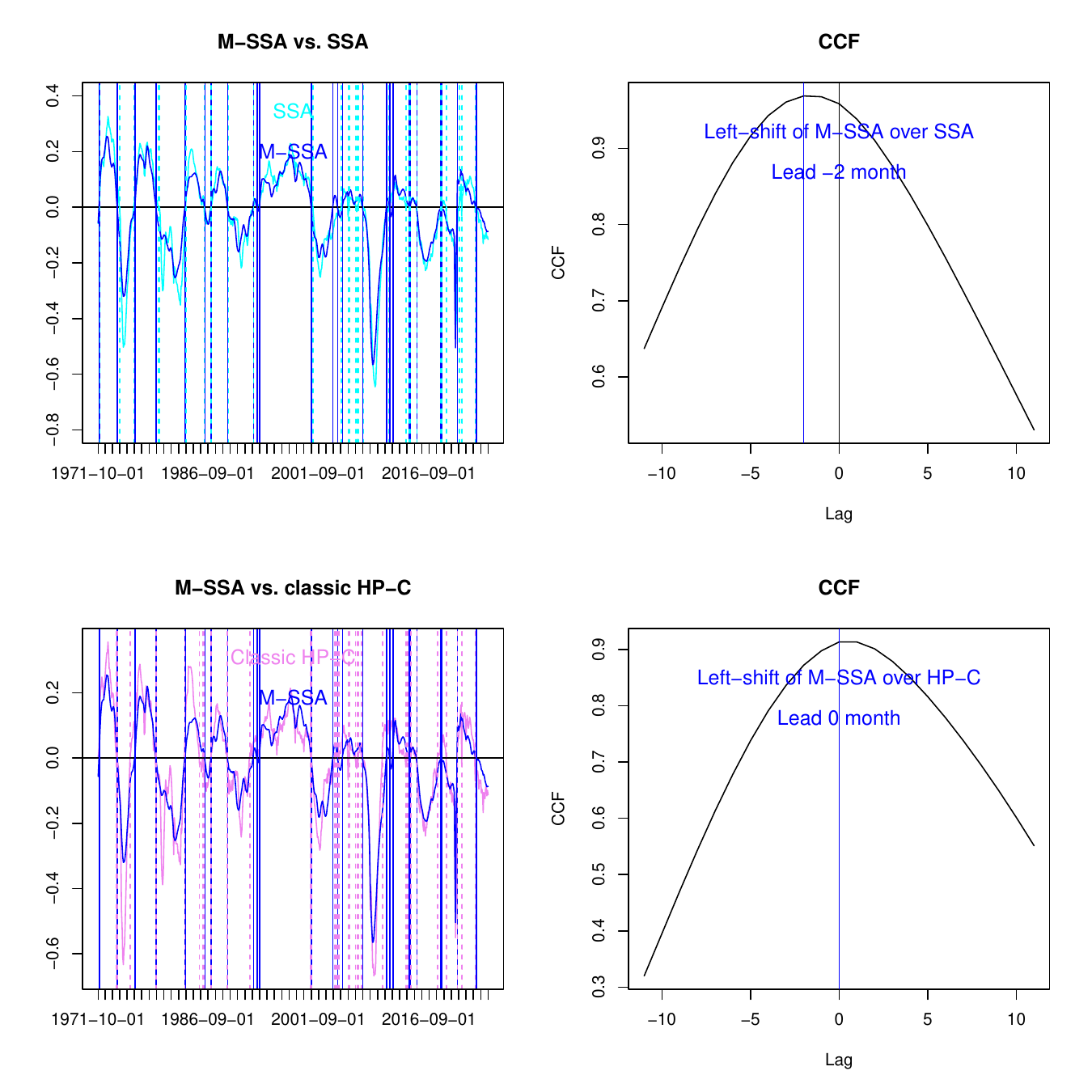}\caption{INDPRO nowcasts: indicators (left) and CCF (right); M-SSA (blue) vs. SSA (cyan) (top left) and M-SSA (blue) vs. classic HP concurrent (violet) (bottom left).  Vertical lines indicating zero-crossings are color-coded accordingly. The CCFs on the right demonstrate that M-SSA leads  the (univariate) SSA and is coincident with the classic HP concurrent filter.\label{comparison_bivariatecli_univariate}}\end{center}\end{figure}To conclude, it is interesting to examine the filtered components of the INDPRO M-SSA nowcast, which are derived from the sub-filters applied to both INDPRO and CLI. Figure \eqref{components_m_ssa} illustrates these components alongside the aggregate for both the financial crisis and the `great lockdown'. Prior to and following the crises, the dynamics of the CLI component are relatively weak, and the nowcast (in blue) aligns with the INDPRO component (in green). However, during the crises, the CLI (in red) exhibits increased dynamics, leading to a divergence of the aggregate M-SSA from the INDPRO component. Notably, the lead introduced by the CLI results in a leftward shift of the aggregate nowcast in relation to the INDPRO component.
\begin{figure}[H]\begin{center}\includegraphics[height=3in, width=6in]{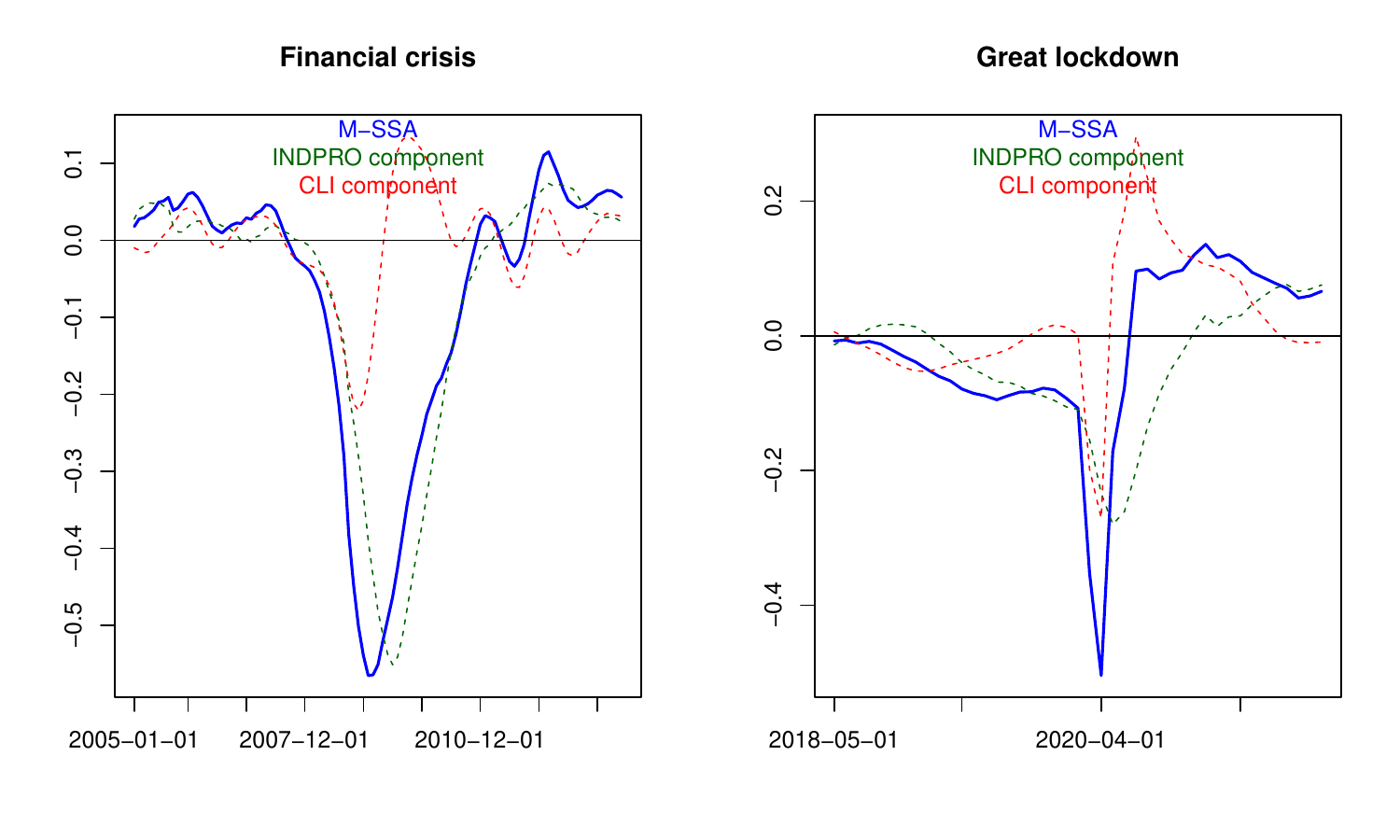}\caption{The INDPRO M-SSA nowcast is represented by the blue solid line, along with its components derived from the sub-filters applied to INDPRO (green dashed line) and the filtered CLI (red dashed line): financial crisis (left) and `great lockdown' (right). The M-SSA nowcast is the aggregate or sum of the two components. \label{components_m_ssa}}\end{center}\end{figure}

\section{Conclusion}\label{conclusion}

We propose a multivariate extension of the SSA, termed M-SSA, as introduced in Wildi (2024). M-SSA represents an innovative approach to prediction and smoothing, highlighting target correlation and the sign accuracy of the predictor while adhering to a novel HT constraint. The conventional MSE approach corresponds to unconstrained M-SSA optimization, subject to a specific scaling factor. However, our proposed criterion captures more nuanced performance metrics related to accuracy and smoothness (AS) characteristics.\\

In its primal formulation, M-SSA seeks to optimally track the target while enforcing noise suppression; conversely, in its dual formulation, the predictor minimizes zero-crossings for a predetermined level of tracking accuracy. A single hyperparameter controls the AS trade-off, tracing a new AS efficient frontier. Whereas the MSE predictor occupies a single spot on this curve, our approach spans the entire frontier.\\

The M-SSA predictor is both interpretable and attractive due to its inherent simplicity, as it integrates fundamental concepts of prediction, such as sign accuracy, mean squared error, and smoothing requirements. The predictor operates under a simple time-reversible unstable difference equation. Importantly, the stability of the predictor is maintained through the presence of implicit `zero' boundary conditions.\\

The M-SSA framework can effectively address forecasting, smoothing, and signal extraction applications, depending on whether its target specification is causal or non-causal, as well as whether it is all-pass or not. This methodology enables the customization of benchmark predictors in terms of smoothness and accuracy performance. Our applications in forecasting, smoothing, and signal extraction demonstrate the customization of MSE and HP benchmarks. Furthermore, our M-SSA package expands this analysis to incorporate additional contemporary business-cycle tools. Notably, the bivariate signal extraction application highlights the significance of leading indicators within a nowcasting framework.

\section{Appendix}

The derivation of the `true' performance metrics, presented in the first row of the table, is briefly elucidated herein. The expected values are derived from the MA inversions of the filters, as outlined in Section \eqref{ext_stat}. Specifically, the target correlation and first-order ACF of the M-SSA are derived from the expressions in Criterion \eqref{gen_stat_x}, relying on the MA-inversion $\boldsymbol{\Xi}_k$, for $k = 0, \ldots, L-1$, associated with the VARMA process. The corresponding MA weights for the first target, INDPRO, are depicted in the top left panel of Fig. \eqref{MA_inv_VARMA}. The target correlations and first-order ACFs for the bivariate MSE and the univariate HP-C benchmarks are obtained using analogous proceedings. In the case of HP-C, the original M-SSA coefficients $\mathbf{B}_k$ are substituted with $\mathbf{HP}_k^{\text{concurrent}} = \begin{pmatrix} hp_k^{\text{concurrent}} & 0 \\ 0 & hp_k^{\text{concurrent}} \end{pmatrix} $, for $k = 0, \ldots, L-1 $ in the expressions for the target correlation and the first-order ACF, where $ hp_k^{\text{concurrent}}$ represents the weights of the classic one-sided HP-C. The zero off-diagonal elements of the one-sided $\mathbf{HP}_k^{\text{concurrent}}$, $k\geq 0$, indicate that each target is influenced solely by its own time series (univariate design). The bivariate MSE predictor can be formulated as described by McElroy and Wildi (2020). Specifically, the expected values presented in the table are contingent upon the MA inversion of the bivariate MSE nowcast, which is conducted in two distinct stages. First, the convolution of $\boldsymbol{\Xi}_k$ with the (truncated) two-sided HP weights $\mathbf{HP}_k^{\text{two-sided}} = \begin{pmatrix} hp_k^{\text{two-sided}} & 0 \\ 0 & hp_k^{\text{two-sided}} \end{pmatrix}$, for $k = -(L-1), \ldots, L-1$, is computed. Second, all weights associated with future innovations $\boldsymbol{\epsilon}_{t-k}$, for $k < 0$, are excluded from this convolution, thereby retaining  only the weights relevant to the causal component of the inverted acausal target.
The resultant MA-inverted MSE predictor is then incorporated into the expressions of Criterion  \eqref{gen_stat_x} (in place of $(\mathbf{b} \cdot \boldsymbol{\xi})_i$) to derive the corresponding first-order ACF and target correlation. To finish, the HTs reported in Table \eqref{tab_bi_0} are derived from Equation \eqref{ht}, utilizing the previously calculated ACFs. The effective predictors, such as those illustrated in the top right and bottom panels of Fig. \eqref{MA_inv_VARMA}, are obtained through the deconvolution of the MA-inverted designs, as described in Equation \eqref{con_inv}. The M-SSA package provides comprehensive guidance on the sequence of operations and transformations necessary to replicate all results, including those presented in the aforementioned table. 
Note also that we exclude the univariate MSE from the benchmark comparisons since it is outperformed by the bi-variate (MSE) design.

\end{document}